%% file: clique_reconf.tex
\documentclass{llncs}

\usepackage{amsmath,amssymb}
\usepackage[dvipdf]{graphicx}

\newcommand{\onestep}{\leftrightarrow}
\newcommand{\sevstep}{\leftrightsquigarrow}

\newcommand{\MC}[2]{\mathsf{MC}_{#1}(#2)} 

\newcommand{\TAR}[1]{\mathsf{TAR}(#1)}
\newcommand{\TS}{\mathsf{TS}}
\newcommand{\TJ}{\mathsf{TJ}}

\newcommand{\symdiff}[2]{#1 \vartriangle #2}

\newcommand{\ini}{0}
\newcommand{\tar}{r}
\newcommand{\cliq}{C}

\newcommand{\TARrule}{\mathsf{TAR}}

\newcommand{\YES}{\mathsf{yes}}
\newcommand{\NO}{\mathsf{no}}

\newcommand{\TARins}[3]{\mathsf{TAR}(#1,#2,#3)}
\newcommand{\TSins}[2]{\mathsf{TS}(#1,#2)}
\newcommand{\TJins}[2]{\mathsf{TJ}(#1,#2)}

\newcommand{\distTAR}[3]{\mathsf{dist_{TAR}}(#1,#2,#3)}
\newcommand{\distTJ}[2]{\mathsf{dist_{TJ}}(#1,#2)}
\newcommand{\distTS}[2]{\mathsf{dist_{TS}}(#1,#2)}

\newcommand{\distTARG}[4]{\mathsf{dist_{TAR}}(#1,#2,#3,#4)}
\newcommand{\distTJG}[3]{\mathsf{dist_{TJ}}(#1,#2,#3)}
\newcommand{\distTSG}[3]{\mathsf{dist_{TS}}(#1,#2,#3)}

\newcommand{\cgraph}{\mathcal{R}}
\newcommand{\cvertex}{\mathcal{V}}
\newcommand{\cedge}{\mathcal{E}}

\newcommand{\Mset}[1]{\mathcal{M}(#1)}
\newcommand{\Msetv}[2]{\mathcal{M}(#1; #2)}

\newcommand{\subH}{H}
\newcommand{\subHp}{H^\prime}
\newcommand{\intH}{\subH}

\newcommand{\dumo}{d_{\ini}}
\newcommand{\dumt}{d_{\tar}}

\newenvironment{listing}[1]{%
        \begin{list}{*}{%
                 \settowidth{\labelwidth}{#1}%
                 \setlength{\leftmargin}{\labelwidth}%
                 \advance \leftmargin by 12pt
                   \setlength{\itemsep}{0pt}%
                   \setlength{\parsep}{0pt}%
                   \setlength{\topsep}{0pt}%
                   \setlength{\parskip}{0pt}%
}%
}{%
\end{list}}

\newcounter{one}
\newcommand{\one}{{\rm \roman{one}}}
\newcounter{two}
\newcommand{\two}{{\rm \roman{two}}}
\newcounter{three}
\newcommand{\three}{{\rm \roman{three}}}
\newcounter{four}
\newcounter{five}
\begin{document}
\title{Reconfiguration of Cliques in a Graph}

\author{
Takehiro Ito\inst{1} \and 
Hirotaka Ono\inst{2} \and
Yota Otachi\inst{3}
}

\institute{
	Graduate School of Information Sciences, 
	Tohoku University, \\
    Aoba-yama 6-6-05, Sendai, 980-8579, Japan.\\
	\email{takehiro@ecei.tohoku.ac.jp}
\and	
	Faculty of Economics, 
	Kyushu University, \\
	Hakozaki 6-19-1, Higashi-ku, Fukuoka, 812-8581, 
	Japan.\\
	\email{hirotaka@econ.kyushu-u.ac.jp}
\and
    School of Information Science, JAIST, \\
    Asahidai 1-1, Nomi, Ishikawa 923-1292, Japan.\\
    \email{otachi@jaist.ac.jp}
}

\maketitle

	\begin{abstract}
	We study reconfiguration problems for cliques in a graph, which determine whether there exists a sequence of cliques that transforms a given clique into another one in a step-by-step fashion. 
	As one step of a transformation, we consider three different types of rules, which are defined and studied in reconfiguration problems for independent sets.
	We first prove that all the three rules are equivalent in cliques.
	We then show that the problems are PSPACE-complete for perfect graphs, while we give polynomial-time algorithms for several classes of graphs, such as even-hole-free graphs and cographs.  
	In particular, the shortest variant, which computes the shortest length of a desired sequence, can be solved in polynomial time for chordal graphs, bipartite graphs, planar graphs, and bounded treewidth graphs.
	\end{abstract}
\vspace{-2em}

	\begin{figure}[b]
	\vspace{-1em}
		\centering
		\includegraphics[width=0.8\linewidth]{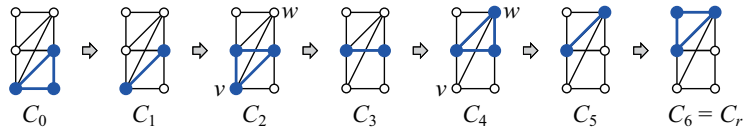}
	\vspace{-1em}
	\caption{A sequence $\langle \cliq_{\ini}, \cliq_1, \ldots, \cliq_6 \rangle$ of cliques in the same graph, where the vertices in cliques are depicted by large (blue) circles (tokens).}
	\label{fig:example}
	\end{figure}

\section{Introduction}
Recently, {\em reconfiguration problems} attract attention 
in the field of theoretical computer science. 
The problem arises when we wish to find a step-by-step transformation between 
two feasible solutions of a problem such that 
all intermediate results are also feasible and 
each step abides by a fixed reconfiguration rule 
(i.e., an adjacency relation defined on feasible solutions of the original problem).
This kind of reconfiguration problem has been studied extensively 
for several well-known problems, including 
{\sc satisfiability}~\cite{Kolaitis}, 
{\sc independent set}~\cite{Bon14,HearnDemaine2005,IDHPSUU,KaminskiMM12,Wro14}, 
{\sc vertex cover}~\cite{INZ14,MNR14},
{\sc clique}, {\sc matching}~\cite{IDHPSUU}, 
{\sc vertex-coloring}~\cite{BC09},
and so on.
(See also a recent survey~\cite{van13}.)

	It is well known that independent sets, vertex covers and cliques are related with each other. 
	Indeed, the well-known reductions for NP-completeness proofs are essentially the same for the three problems~\cite{GJ79}.
	Despite reconfiguration problems for independent sets and vertex covers are two of the most well studied problems, we have only a few known results for reconfiguration problems for cliques (as we will explain later). 
	In this paper, we thus investigate the complexity status of reconfiguration problems for cliques systematically, and show that the problems can be solved in polynomial time for a variety of graph classes, in contrast to independent sets and vertex covers.  

	\subsection{Our problems and three rules}

	Recall that a {\em clique} of a graph $G = (V,E)$ is a vertex subset of $G$ in which every two vertices are adjacent. 
(Figure~\ref{fig:example} depicts seven different cliques in the same graph.)
	Suppose that we are given two cliques $\cliq_{\ini}$ and $\cliq_{\tar}$ of $G$, and imagine that a token is placed on each vertex in $\cliq_{\ini}$. 
	Then, we are asked to transform $\cliq_{\ini}$ into $\cliq_{\tar}$ by abiding a prescribed reconfiguration rule on cliques.
	In this paper, we define three different reconfiguration rules on cliques, which were originally defined as the reconfiguration rules on independents sets~\cite{KaminskiMM12}, as follows:
\begin{list}{*}{%
	\settowidth{\labelwidth}{$\bullet$}%
	\setlength{\leftmargin}{\labelwidth}%
	\advance \leftmargin by 5pt
	\setlength{\itemsep}{5pt}%
	\setlength{\parsep}{0pt}%
	\setlength{\topsep}{5pt}%
	\setlength{\parskip}{0pt}%
}
	\item[$\bullet$] {\em Token Addition and Removal} ($\TARrule$ rule): 
    We can either add or remove a single token at a time 
    if it results in a clique of size at least a given threshold $k \ge 0$. 
    For example, in the sequence $\langle \cliq_{\ini}, \cliq_1, \ldots, \cliq_6 \rangle$ in \figurename~\ref{fig:example}, every two consecutive cliques follow the $\TARrule$ rule for the threshold $k = 2$.
    In order to emphasize the threshold $k$, we sometimes call this rule the $\TAR{k}$ rule. 

	\item[$\bullet$] {\em Token Jumping} ($\TJ$ rule): 
    A single token in a clique $\cliq$ can ``jump'' to any vertex in $V \setminus \cliq$ if it results in a clique.
    For example, consider the sequence $\langle \cliq_{\ini}, \cliq_2, \cliq_4, \cliq_6 \rangle$ in \figurename~\ref{fig:example}, then two consecutive cliques $\cliq_{2i}$ and $\cliq_{2i+2}$ follow the $\TJ$ rule for each $i \in \{0,1,2\}$.

	\item[$\bullet$]{\em Token Sliding} ($\TS$ rule): 
    We can slide a single token on a vertex $v$ in a clique $\cliq$ to another vertex $w$ in $V \setminus \cliq$ if  it results in a clique and there is an edge $vw$ in $G$.
    For example, consider the sequence $\langle \cliq_2, \cliq_4 \rangle$ in \figurename~\ref{fig:example}, then two consecutive cliques $\cliq_{2}$ and $\cliq_{4}$ follow the $\TS$ rule, because $v$ and $w$ are adjacent.
	\end{list}
	A sequence $\langle \cliq_{\ini}, \cliq_1, \ldots, \cliq_{\ell} \rangle$ of cliques of a graph $G$ is called a {\em reconfiguration sequence} between two cliques $\cliq_{\ini}$ and $\cliq_{\ell}$ under $\TAR{k}$ (or $\TJ$, $\TS$) if two consecutive cliques $\cliq_{i-1}$ and $\cliq_i$ follow the $\TAR{k}$ (resp., $\TJ$, $\TS$) rule for all $i \in \{1, 2, \ldots, \ell\}$.
	The \emph{length} of a reconfiguration sequence is defined to be the number of cliques in the sequence minus one, that is, the length of  $\langle \cliq_{\ini}, \cliq_1, \ldots, \cliq_{\ell} \rangle$ is $\ell$.

	Given two cliques $\cliq_{\ini}$ and $\cliq_{\tar}$ of a graph $G$ (and an integer $k \ge 0$ for $\TARrule$), \textsc{clique reconfiguration} under $\TARrule$ (or $\TJ$, $\TS$) is to determine whether there exists a reconfiguration sequence between $\cliq_{\ini}$ and $\cliq_{\tar}$ under $\TAR{k}$ (resp., $\TJ$, $\TS$). 
	For example, consider the cliques $\cliq_{\ini}$ and $\cliq_{\tar} = \cliq_6$ in \figurename~\ref{fig:example}; let $k = 2$ for $\TARrule$. 
	Then, it is a $\YES$-instance under the $\TAR{2}$ and $\TJ$ rules as illustrated in \figurename~\ref{fig:example}, but is a $\NO$-instance under the $\TS$ rule.

	In this paper, we also study the shortest variant, called \textsc{shortest clique reconfiguration}, under each of the three rules which computes the shortest length of a reconfiguration sequence between two given cliques under the rule.
	We define the shortest length to be infinity for a $\NO$-instance, and hence this variant is a generalization of \textsc{clique reconfiguration}. 

	\subsection{Known and related results}
	
	Ito et al.~\cite{IDHPSUU} introduced \textsc{clique reconfiguration} under $\TARrule$, and proved that it is PSPACE-complete in general. 
	They also considered the optimization problem of computing the maximum threshold $k$ such that there is a reconfiguration sequence between two given cliques $\cliq_{\ini}$ and $\cliq_{\tar}$ under $\TAR{k}$. 
	This maximization problem cannot be approximated in polynomial time within any constant factor unless ${\rm P} = {\rm NP}$~\cite{IDHPSUU}.

	\textsc{Independent set reconfiguration} is one of the most well-studied reconfiguration problems, defined for independent sets in a graph. 
	Kami\'nski et al.~\cite{KaminskiMM12} studied the problem under $\TARrule$, $\TJ$ and $\TS$. 
	It is well known that a clique in a graph $G$ forms an independent set in the complement $\overline{G}$ of $G$, and vice versa.
	Indeed, some known results for \textsc{independent set reconfiguration} can be converted into ones for \textsc{clique reconfiguration}.
	However, as far as we checked, only two results can be obtained for \textsc{clique reconfiguration} by this conversion, because we take the complement of a graph.
(These results will be formally discussed in Section~\ref{subsec:independent-clique}.)

	In this way, only a few results are known for \textsc{clique reconfiguration}. 
	In particular, there is almost no algorithmic result, and hence it is desired to develop efficient algorithms for the problem and its shortest variant. 

	\subsection{Our contribution}

	In this paper, we embark on a systematic investigation of the computational status of \textsc{clique reconfiguration} and its shortest variant. 
	Figure~\ref{fig:results} summarizes our results, which can be divided into the following four parts.
\smallskip

\begin{list}{*}{%
	\settowidth{\labelwidth}{(3)}%
	\setlength{\leftmargin}{\labelwidth}%
	\advance \leftmargin by 5pt
	\setlength{\itemsep}{5pt}%
	\setlength{\parsep}{0pt}%
	\setlength{\topsep}{0pt}%
	\setlength{\parskip}{0pt}%
}
	\item[(1)] \emph{Rule equivalence} (Section~\ref{sec:rules}): We prove that  all rules $\TARrule$, $\TS$ and $\TJ$ are equivalent in \textsc{clique reconfiguration}.
	Then, any complexity result under one rule can be converted into the same complexity result under the other two rules. 
	In addition, based on the rule equivalence, we show that \textsc{clique reconfiguration} under any rule is PSPACE-complete for perfect graphs, and is solvable in linear time for cographs.

	\item[(2)] \emph{Graphs with bounded clique size} (Section~\ref{subsec:boundedclique}): We show that the shortest variant under any of $\TARrule$, $\TS$ and $\TJ$ can be solved in polynomial time for such graphs, which include bipartite graphs, planar graphs, and bounded treewidth graphs.
	Interestingly, \textsc{independent set reconfiguration} under any rule remains PSPACE-complete even for planar graphs~\cite{BC09,HearnDemaine2005} and bounded treewidth graphs~\cite{Wro14}. 
	Therefore, this result shows a nice difference between the reconfiguration problems for cliques and independent sets. 

	\item[(3)] \emph{Graphs with polynomially many maximal cliques} (Section~\ref{subsec:polymany}): We  show that \textsc{clique reconfiguration} under any of $\TARrule$, $\TS$ and $\TJ$ can be solved in polynomial time for such graphs, which include even-hole-free graphs, graphs of bounded boxicity, and $K_{t}$-subdivision-free graphs.
	
	\item[(4)] \emph{Chordal graphs} (Section~\ref{sec:chordal}): We give a linear-time algorithm to solve the shortest variant under any of $\TARrule$, $\TS$ and $\TJ$ for chordal graphs.
	Note that the clique size of chordal graphs is not always bounded, and hence this result is independent from Result (2) above. 
\end{list}
	Several proofs move to appendices.

	\begin{figure}[t]
		\centering
		\includegraphics[width=0.9\linewidth]{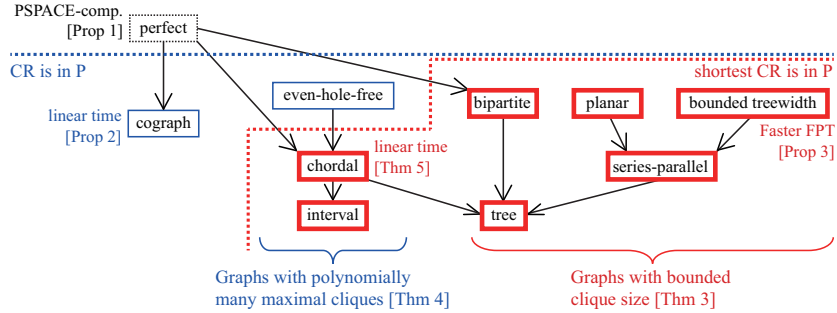}
	\vspace{-1em}
	\caption{Our results under all rules $\TARrule$, $\TS$ and $\TJ$. 
Each arrow represents the inclusion relationship between graph classes: 
$A \to B$ represents that $B$ is properly included in $A$~\cite{BLS99}.
Graph classes for which \textsc{shortest clique reconfiguration} is solvable in polynomial time are indicated by thick (red) boxes, 
while the ones for which \textsc{clique reconfiguration} is solvable in polynomial time are indicated by thin (blue) boxes.}
	\vspace{-1em}
	\label{fig:results}
	\end{figure}


	\section{Preliminaries}
	In this section, we introduce some basic terms and notation. 

	\subsection{Graph notation}

	In this paper, we assume without loss of generality that graphs are simple.
	For a graph $G$, we sometimes denote by $V(G)$ and $E(G)$ the vertex set and edge set of $G$, respectively. 
%
%
	For a graph $G$, the \emph{complement} $\overline{G}$ of $G$ is the graph such that $V(\overline{G}) = V(G)$ and $E(\overline{G}) = \{ vw \mid v,w \in V(G),\ vw \not\in E(G) \}$.
	We say that a graph class $\mathcal{G}$ (i.e., a set of graphs) is \emph{closed under taking complements} if $\overline{G} \in \mathcal{G}$ holds for every graph $G \in \mathcal{G}$.

	In this paper, we deal with several graph classes systematically, and hence we do not define those graph classes precisely;
we simply give the properties used for proving  our results, with appropriate references. 

	\subsection{Definitions for {\sc clique reconfiguration}}

	As explained in Introduction, we consider three (symmetric) adjacency relations on cliques in a graph.
	Let $\cliq_i$ and $\cliq_j$ be two cliques of a graph $G$.
	Then,
	\begin{listing}{a}
	\item[$\bullet$] \emph{$\cliq_i \onestep \cliq_j$ under $\TAR{k}$} for a nonnegative integer $k$ if $|\cliq_i| \ge k$, $|\cliq_j| \ge k$, and $|\cliq_i \vartriangle \cliq_j| = \bigl|(\cliq_i \setminus \cliq_j) \cup (\cliq_j \setminus \cliq_i) \bigr| = 1$ hold;
	\smallskip

	\item[$\bullet$] \emph{$\cliq_i \onestep \cliq_j$ under $\TJ$} if $|\cliq_i| = |\cliq_j|$, $|\cliq_i \setminus \cliq_j| = 1$, and $|\cliq_j \setminus \cliq_i| = 1$ hold; and
	\smallskip

	\item[$\bullet$] \emph{$\cliq_i \onestep \cliq_j$ under $\TS$} if $|\cliq_i| = |\cliq_j|$, $\cliq_i \setminus \cliq_j = \{v\}$, $\cliq_j \setminus \cliq_i = \{w\}$, and $vw \in E(G)$ hold.
	\smallskip
	\end{listing}
	A sequence $\langle \cliq_1, \cliq_2, \ldots, \cliq_{\ell} \rangle$ of cliques of $G$ is called a {\em reconfiguration sequence} between two cliques $\cliq_1$ and $\cliq_{\ell}$ under $\TAR{k}$ (or $\TJ$, $\TS$) if $\cliq_{i-1} \onestep \cliq_i$ holds under $\TAR{k}$ (resp., $\TJ$, $\TS$) for all $i \in \{2, 3, \ldots, \ell\}$.
	A reconfiguration sequence under $\TAR{k}$ (or $\TJ$, $\TS$) is simply called a \emph{$\TAR{k}$-sequence} (resp., \emph{$\TJ$-sequence}, \emph{$\TS$-sequence}).
	We write $\cliq_{1} \sevstep \cliq_{\ell}$ under $\TAR{k}$ (or $\TJ$, $\TS$) if there exists a $\TAR{k}$-sequence (resp., $\TJ$-sequence, $\TS$-sequence) between $\cliq_1$ and $\cliq_{\ell}$.
	Note that each clique in any $\TAR{k}$-sequence is of size at least $k$, while all cliques in any $\TJ$-sequence or $\TS$-sequence have the same size.
	In addition, a reconfiguration sequence under any rule is {\em reversible}, that is, $\cliq_{1} \sevstep \cliq_{\ell}$ if and only if $\cliq_{\ell} \sevstep \cliq_{1}$.

	Let $k$ be a nonnegative integer, and let $\cliq$ and $\cliq^\prime$ be two cliques of a graph $G$.
	Then, we define $\TARins{\cliq}{\cliq^\prime}{k}$, as follows:
	\[
		\TARins{\cliq}{\cliq^\prime}{k} = \left\{
				\begin{array}{ll}
				\YES & ~~~\mbox{if $\cliq \sevstep \cliq^\prime$ under $\TAR{k}$}; \\
				\NO  & ~~~\mbox{otherwise}.
				\end{array} \right.
	\]
	Given two cliques $\cliq_{\ini}$ and $\cliq_{\tar}$ of a graph $G$ and a nonnegative integer $k$, \textsc{clique reconfiguration} under $\TARrule$ is to compute $\TARins{\cliq_{\ini}}{\cliq_{\tar}}{k}$.
	By the definition, $\TARins{\cliq_{\ini}}{\cliq_{\tar}}{k} = \NO$ if $|\cliq_{\ini}| < k$ or $|\cliq_{\tar}| <k$ hold, and hence we may assume without loss of generality that both $|\cliq_{\ini}| \ge k$ and $|\cliq_{\tar}| \ge k$ hold;
we call such an instance simply a $\TARrule$-instance, and denote it by $(G, \cliq_{\ini}, \cliq_{\tar}, k)$.
	
	For two cliques $\cliq$ and $\cliq^\prime$ of a graph $G$, we similarly define $\TJins{\cliq}{\cliq^\prime}$ and $\TSins{\cliq}{\cliq^\prime}$.
	Given two cliques $\cliq_{\ini}$ and $\cliq_{\tar}$ of $G$, we similarly define \textsc{clique reconfiguration} under $\TJ$ and $\TS$, and denote their instance by $(G, \cliq_{\ini}, \cliq_{\tar})$.
	Then, we can assume that $|\cliq_{\ini}| = |\cliq_{\tar}|$ holds in a $\TJ$- or a $\TS$-instance $(G, \cliq_{\ini}, \cliq_{\tar})$.

	Given a $\TARrule$-instance $(G, \cliq_{\ini}, \cliq_{\tar}, k)$, let $\mathcal{C} = \langle \cliq_{\ini}, \cliq_1, \ldots, \cliq_{\ell} \rangle$ be a $\TAR{k}$-sequence in $G$ between $\cliq_{\ini}$ and $\cliq_{\tar} = \cliq_{\ell}$. 
	Then, the \emph{length} of $\mathcal{C}$ is defined to be the number of cliques in $\mathcal{C}$ minus one, that is, the length of $\mathcal{C}$ is $\ell$.
	We denote by $\distTARG{G}{\cliq_{\ini}}{\cliq_{\tar}}{k}$ the minimum length of a $\TAR{k}$-sequence in $G$ between $\cliq_{\ini}$ and $\cliq_{\tar}$;
we let $\distTARG{G}{\cliq_{\ini}}{\cliq_{\tar}}{k} = +\infty$ if there is no $\TAR{k}$-sequence in $G$ between $\cliq_{\ini}$ and $\cliq_{\tar}$. 
	The shortest variant, \textsc{shortest clique reconfiguration}, under $\TARrule$ is to compute $\distTARG{G}{\cliq_{\ini}}{\cliq_{\tar}}{k}$. 
	Similarly, we define $\distTJG{G}{\cliq_{\ini}}{\cliq_{\tar}}$ and $\distTSG{G}{\cliq_{\ini}}{\cliq_{\tar}}$ for a $\TJ$- and a $\TS$-instance $(G, \cliq_{\ini}, \cliq_{\tar})$, respectively. 
	Then, \textsc{shortest clique reconfiguration} under $\TJ$ or $\TS$ is defined similarly. 
	We sometimes drop $G$ and simply write $\distTAR{\cliq_{\ini}}{\cliq_{\tar}}{k}$, $\distTJ{\cliq_{\ini}}{\cliq_{\tar}}$ and $\distTS{\cliq_{\ini}}{\cliq_{\tar}}$ if it is clear from context.
	
	We note that \textsc{clique reconfiguration} under any rule is a decision problem asking for the existence of a reconfiguration sequence, and its shortest variant asks for simply computing the shortest length of a reconfiguration sequence. 
	Therefore, the problems do not ask for an actual reconfiguration sequence. 
	However, our algorithms proposed in this paper can be easily modified so that they indeed find a reconfiguration sequence.  


\section{Rule Equivalence and Complexity}
\label{sec:rules}

	In this section, we first prove that all three rules $\TARrule$, $\TS$ and $\TJ$ are equivalent in \textsc{clique reconfiguration}.
	We then discuss some complexity results that can be obtained from known results for \textsc{independent set reconfiguration}.
	

	\subsection{Equivalence of $\TS$ and $\TARrule$ rules} 
	
	$\TS$ and $\TARrule$ rules are equivalent, as in the following sense.
	\begin{theorem} \label{the:TS=TAR}
	$\TS$ and $\TARrule$ rules are equivalent in \textsc{clique reconfiguration}, as follows{\rm :} 
		\begin{listing}{aaa}
		\item[{\rm (}a{\rm )}] for any $\TS$-instance $(G, \cliq_{\ini}, \cliq_{\tar})$, a $\TARrule$-instance $(G, \cliq_{\ini}^\prime, \cliq_{\tar}^\prime, k^\prime)$ can be constructed in linear time such that  $\TSins{\cliq_{\ini}}{\cliq_{\tar}} = \TARins{\cliq_{\ini}^\prime}{\cliq_{\tar}^\prime}{k^\prime}$ and $\distTS{\cliq_{\ini}}{\cliq_{\tar}} = \distTAR{\cliq_{\ini}^\prime}{\cliq_{\tar}^\prime}{k^\prime} /2${\rm ;} and
		\item[{\rm (}b{\rm )}] for any $\TARrule$-instance $(G, \cliq_{\ini}, \cliq_{\tar},k)$, a $\TS$-instance $(G, \cliq_{\ini}^\prime, \cliq_{\tar}^\prime)$ can be constructed in linear time such that $\TARins{\cliq_{\ini}}{\cliq_{\tar}}{k} = \TSins{\cliq_{\ini}^\prime}{\cliq_{\tar}^\prime}$.
		\end{listing}
	\end{theorem}
	By Theorem~\ref{the:TS=TAR}(a), note that the reduction from $\TS$ to $\TARrule$ preserves the shortest length of reconfiguration sequences.
\medskip

	\noindent
	{\em Proof of Theorem~{\rm \ref{the:TS=TAR}(}a{\rm )}.}
	Let $(G, \cliq_{\ini}, \cliq_{\tar})$ be a $\TS$-instance with $|\cliq_{\ini}| = |\cliq_{\tar}| = k$.
	Then, as the corresponding $\TARrule$-instance $(G, \cliq_{\ini}^\prime, \cliq_{\tar}^\prime, k^\prime)$, we let $\cliq_{\ini}^\prime = \cliq_{\ini}$, $\cliq_{\tar}^\prime = \cliq_{\tar}$ and $k^\prime = k$;
this $\TARrule$-instance can be clearly constructed in linear time.
	We thus prove the following lemma, as a proof of Theorem~\ref{the:TS=TAR}(a).
	\begin{lemma} \label{lem:TS->TAR}
	Let $G$ be a graph, and let $\cliq_{\ini}$ and $\cliq_{\tar}$ be any pair of cliques of $G$ such that $|\cliq_{\ini}| = |\cliq_{\tar}| = k$.
	Then, $\TSins{\cliq_{\ini}}{\cliq_{\tar}} = \TARins{\cliq_{\ini}}{\cliq_{\tar}}{k}$ and $\distTS{\cliq_{\ini}}{\cliq_{\tar}} = \distTAR{\cliq_{\ini}}{\cliq_{\tar}}{k} /2$.
	\end{lemma}

	\noindent
	{\em Proof of Theorem~{\rm \ref{the:TS=TAR}(}b{\rm )}.}
	Let $(G, \cliq_{\ini}, \cliq_{\tar}, k)$ be a $\TARrule$-instance; 
note that $|\cliq_{\ini}| \neq |\cliq_{\tar}|$ may hold, and both $|\cliq_{\ini}| \ge k$ and $|\cliq_{\tar}| \ge k$ hold.
	Then, as the corresponding $\TS$-instance $(G, \cliq_{\ini}^\prime, \cliq_{\tar}^\prime)$, let $\cliq_{\ini}^\prime \subseteq \cliq_{\ini}$ and $\cliq_{\tar}^\prime \subseteq \cliq_{\tar}$ be arbitrary subsets of size exactly $k$;
this $\TS$-instance can be clearly constructed in linear time.
	We thus prove the following lemma, as a proof of Theorem~\ref{the:TS=TAR}(b).
	\begin{lemma} \label{lem:TS=TARb}
	Let $(G, \cliq_{\ini}, \cliq_{\tar}, k)$ be a $\TARrule$-instance, and let $\cliq_{\ini}^\prime \subseteq \cliq_{\ini}$ and $\cliq_{\tar}^\prime \subseteq \cliq_{\tar}$ be arbitrary subsets of size exactly $k$.
	Then, $\TARins{\cliq_{\ini}}{\cliq_{\tar}}{k} = \TSins{\cliq_{\ini}^\prime}{\cliq_{\tar}^\prime}$.
	\end{lemma}

	\subsection{Equivalence of $\TJ$ and $\TARrule$ rules} 
	
	$\TJ$ and $\TARrule$ rules are equivalent, as in the following sense.
	\begin{theorem} \label{the:TJ=TAR}
	$\TJ$ and $\TARrule$ rules are equivalent in \textsc{clique reconfiguration}, as follows{\rm :} 
		\begin{listing}{aaa}
		\item[{\rm (}a{\rm )}] for any $\TJ$-instance $(G, \cliq_{\ini}, \cliq_{\tar})$, a $\TARrule$-instance $(G, \cliq_{\ini}^\prime, \cliq_{\tar}^\prime, k^\prime)$ can be constructed in linear time such that  $\TJins{\cliq_{\ini}}{\cliq_{\tar}} = \TARins{\cliq_{\ini}^\prime}{\cliq_{\tar}^\prime}{k^\prime}$ and $\distTJ{\cliq_{\ini}}{\cliq_{\tar}} = \distTAR{\cliq_{\ini}^\prime}{\cliq_{\tar}^\prime}{k^\prime} /2${\rm ;} and
		\item[{\rm (}b{\rm )}] for any $\TARrule$-instance $(G, \cliq_{\ini}, \cliq_{\tar},k)$, a $\TJ$-instance $(G, \cliq_{\ini}^\prime, \cliq_{\tar}^\prime)$ can be constructed in linear time such that $\TARins{\cliq_{\ini}}{\cliq_{\tar}}{k} = \TJins{\cliq_{\ini}^\prime}{\cliq_{\tar}^\prime}$.
		\end{listing}
	\end{theorem}
	By Theorem~\ref{the:TJ=TAR}(a), note that the reduction from $\TJ$ to $\TARrule$ preserves the shortest length of reconfiguration sequences.
%
\medskip

	\noindent
	{\em Proof of Theorem~{\rm \ref{the:TJ=TAR}(}a{\rm )}.}
	Let $(G, \cliq_{\ini}, \cliq_{\tar})$ be a $\TJ$-instance with $|\cliq_{\ini}| = |\cliq_{\tar}| = k$.
	Then, as the corresponding $\TARrule$-instance $(G, \cliq_{\ini}^\prime, \cliq_{\tar}^\prime, k^\prime)$, we let $\cliq_{\ini}^\prime = \cliq_{\ini}$, $\cliq_{\tar}^\prime = \cliq_{\tar}$ and $k^\prime = k-1$;
this $\TARrule$-instance can be clearly constructed in linear time.
	We thus prove the following lemma, as a proof of Theorem~\ref{the:TJ=TAR}(a).
	\begin{lemma} \label{lem:TJ->TAR}
	Let $G$ be a graph, and let $\cliq_{\ini}$ and $\cliq_{\tar}$ be any pair of cliques of $G$ such that $|\cliq_{\ini}| = |\cliq_{\tar}| = k$.
	Then, $\TJins{\cliq_{\ini}}{\cliq_{\tar}} = \TARins{\cliq_{\ini}}{\cliq_{\tar}}{k-1}$ and $\distTJ{\cliq_{\ini}}{\cliq_{\tar}} = \distTAR{\cliq_{\ini}}{\cliq_{\tar}}{k-1} /2$.
	\end{lemma}

	\noindent
	{\em Proof of Theorem~{\rm \ref{the:TJ=TAR}(}b{\rm )}.}
	Let $(G, \cliq_{\ini}, \cliq_{\tar}, k)$ be a $\TARrule$-instance; 
$|\cliq_{\ini}| \neq |\cliq_{\tar}|$ may hold, and both $|\cliq_{\ini}| \ge k$ and $|\cliq_{\tar}| \ge k$ hold.
	We first give the following lemma.
	\begin{lemma}
	Let $(G, \cliq_{\ini}, \cliq_{\tar}, k)$ be a $\TARrule$-instance such that $\cliq_{\ini} \neq \cliq_{\tar}$.
	Suppose that there exists an index $j \in \{\ini, \tar\}$ such that $|\cliq_j| =k$ and $\cliq_j$ is a maximal clique in $G$.
	Then, $\TARins{\cliq_{\ini}}{\cliq_{\tar}}{k} = \NO$. 
	\end{lemma}
	\begin{proof}
	Since $\cliq_j$ is maximal, there is no clique in $G$ which can be obtained by adding a vertex to $\cliq_j$. 
	Furthermore, since $|\cliq_j| = k$, we cannot delete any vertex from $\cliq_j$ to keep the threshold $k$. 
	Thus, there is no clique $C$ in $G$ such that $\cliq_j \onestep C$ under $\TAR{k}$.
	Since $\cliq_{\ini} \neq \cliq_{\tar}$, we have $\TARins{\cliq_{\ini}}{\cliq_{\tar}}{k} = \NO$. 
	\qed
	\end{proof}

	We thus assume without loss of generality that none of $\cliq_{\ini}$ and $\cliq_{\tar}$ is a maximal clique in $G$ of size $k$;
note that the maximality of a clique can be determined in linear time. 
	Then, we construct the corresponding $\TJ$-instance $(G, \cliq_{\ini}^\prime, \cliq_{\tar}^\prime)$, as in the following two cases (\one) and (\two):
	\begin{listing}{aaa}
	\item[(\one)] for each $j \in \{\ini, \tar\}$ such that $|\cliq_j| \ge k+1$, let $\cliq_j^\prime \subseteq \cliq_j$ be an arbitrary subset of size exactly $k+1$; and
	\item[(\two)] for each $j \in \{\ini, \tar\}$ such that $|\cliq_j| = k$, let $\cliq_j^\prime \supset \cliq_j$ be an arbitrary superset of size exactly $k+1$.
	\end{listing}
	This $\TJ$-instance can be clearly constructed in linear time.
	We thus prove the following lemma, as a proof of Theorem~\ref{the:TJ=TAR}(b).
	\begin{lemma} \label{lem:TJ=TARb}
	Let $(G, \cliq_{\ini}, \cliq_{\tar}, k)$ be a $\TARrule$-instance, and let $(G, \cliq_{\ini}^\prime, \cliq_{\tar}^\prime)$ be the corresponding $\TJ$-instance constructed above.
	Then, $\TARins{\cliq_{\ini}}{\cliq_{\tar}}{k} = \TJins{\cliq_{\ini}^\prime}{\cliq_{\tar}^\prime}$.
	\end{lemma}

\subsection{Results obtained from {\sc independent set reconfiguration}}
\label{subsec:independent-clique}

	We here show two complexity results for \textsc{clique reconfiguration}, which can be obtained from known results for \textsc{independent set reconfiguration}. 

	Consider a vertex subset $\cliq$ of a graph $G$. 
	Then, $\cliq$ forms a clique in $G$ if and only if $\cliq$ forms an independent set in the complement $\overline{G}$ of $G$.
	Therefore, the following lemma clearly holds.
	\begin{lemma} \label{lem:clique-independent}
	Let $G$ be a graph, and let $\cliq_{j}$ be a clique of $G$ for each $j \in \{0, 1, \ldots, \ell\}$.
	Then, $\langle \cliq_{0}, \cliq_{1}, \ldots, \cliq_{\ell} \rangle$ is a $\TAR{k}$-sequence of cliques in $G$
if and only if
$\langle \cliq_{0}, \cliq_{1}, \allowbreak \ldots, \cliq_{\ell} \rangle$ is a $\TAR{k}$-sequence of independent sets in the complement $\overline{G}$ of $G$.
	\end{lemma}

	By Lemma~\ref{lem:clique-independent} we can convert a complexity result for \textsc{independent set reconfiguration} under $\TARrule$ for a graph class $\mathcal{G}$ into one for \textsc{clique reconfiguration} under $\TARrule$ for $\mathcal{G}$ if the graph class $\mathcal{G}$ is closed under taking complements. 
	Note that, by Theorems~\ref{the:TS=TAR} and \ref{the:TJ=TAR}, any complexity result under one rule can be converted into the same complexity result under the other two rules. 

	\begin{proposition} \label{pro:perfect}
	\textsc{Clique reconfiguration} is PSPACE-complete for perfect graphs under all rules $\TARrule$, $\TS$ and $\TJ$.
	\end{proposition}

	\begin{proposition} \label{pro:cograph}
	\textsc{Clique reconfiguration} can be solved in linear time for cographs under all rules $\TARrule$, $\TS$ and $\TJ$.
	\end{proposition}


	\section{Polynomial-Time Algorithms} \label{sec:polytime}
	In this section, we show that \textsc{clique reconfiguration} is solvable in polynomial time for several graph classes. 
	We deal with two types of graph classes, that is, graphs of bounded clique size (in Section~\ref{subsec:boundedclique}) and graphs having polynomially many maximal cliques (in Section~\ref{subsec:polymany}).

	\subsection{Graphs of bounded clique size}
	\label{subsec:boundedclique}
	
	In this subsection, we show that \textsc{shortest clique reconfiguration} can be solved in polynomial time for graphs of bounded clique size; 
as we will explain later, such graphs include bipartite graphs, planar graphs, and graphs of bounded treewidth.
	For a graph $G$, we denote by $\omega(G)$ the size of a maximum clique in $G$.
	Then, we have the following theorem. 
	\begin{theorem} \label{the:boundedclique}
	Let $G$ be a graph with $n$ vertices such that $\omega(G) \le w$ for a positive integer $w$. 
	Then, \textsc{shortest clique reconfiguration} under any of $\TARrule$, $\TS$ and $\TJ$ can be solved in time $O(w^2 n^{w})$ for $G$. 
	\end{theorem}
	
	It is well known that $\omega(G) \le 4$ for any planar graph $G$, and $\omega(G^\prime) \le 2$ for any bipartite graph $G^\prime$.
	We thus have the following corollary. 
	\begin{corollary}
	\textsc{Shortest clique reconfiguration} under $\TARrule$, $\TS$ and $\TJ$ can be solved in polynomial time for planar graphs and bipartite graphs.
	\end{corollary}
	
	By the definition of treewidth~\cite{BodlaenderDDFLP13}, we have $\omega(G) \le t+1$ for any graph $G$ whose treewidth can be bounded by a positive integer $t$.
	By Theorem~\ref{the:boundedclique} this observation gives an $O \bigl(t^2 n^{t+1} \bigr)$-time algorithm for \textsc{shortest clique reconfiguration}.
	However, for this case, we can obtain a faster fixed-parameter algorithm, where the parameter is the treewidth $t$, as follows.
	\begin{proposition} \label{pro:treewidth}
	Let $G$ be a graph with $n$ vertices whose treewidth is bounded by a positive integer $t$.
	Then, \textsc{shortest clique reconfiguration} under any of $\TARrule$, $\TS$ and $\TJ$ can be solved for $G$ in time $O(c^{t} n)$, where $c$ is some constant.
	\end{proposition}

	Proposition~\ref{pro:treewidth} implies that \textsc{shortest clique reconfiguration} under any of $\TARrule$, $\TS$ and $\TJ$ can be solved in time $O(c^{w} n)$ for chordal graphs $G$ when parameterized by the size $w$ of a maximum clique in $G$, where $n$ is the number of vertices in $G$ and $c$ is some constant; because the treewidth of a chordal graph $G$ can be bounded by the size of a maximum clique in $G$ minus one~\cite{RS86}. 
	However, we give a linear-time algorithm to solve the shortest variant under any rule for chordal graphs in Section~\ref{sec:chordal}.


	\subsection{Graphs with polynomially many maximal cliques}
	\label{subsec:polymany}
	
	In this subsection, we consider the class of graphs having polynomially many maximal cliques, which properly contains the class of graphs with bounded clique size (in Section~\ref{subsec:boundedclique}). 
	Note that, even if a graph $G$ has a polynomial number of maximal cliques, $G$ may have a super-polynomial number of cliques. 
	\begin{theorem} \label{the:poly-many-maxcliques}
	Let $G$ be a graph with $n$ vertices and $m$ edges, and let $\Mset{G}$ be the set of all maximal cliques in $G$. 
	Then, \textsc{clique reconfiguration} under any of $\TARrule$, $\TS$ and $\TJ$ can be solved for $G$ in time $O \bigl(m n \lvert \mathcal{M}(G) \rvert + n \lvert \mathcal{M}(G) \rvert^{2} \bigr)$.
	\end{theorem}

	Before proving Theorem~\ref{the:poly-many-maxcliques}, we give the following corollary.
	\begin{corollary} \label{cor:maximal}
	\textsc{Clique reconfiguration} under $\TARrule$, $\TS$ and $\TJ$ can be solved in polynomial time for even-hole-free graphs, graphs of bounded boxicity, and $K_{t}$-subdivision-free graphs.
	\end{corollary}
	\begin{proof}
	By Theorem~\ref{the:poly-many-maxcliques} it suffices to show that the claimed graphs have polynomially many maximal cliques.
	Polynomial bounds on the number of maximal cliques are shown
for even-hole-free graphs in~\cite{SilvaV07},
for graphs of bounded boxicity in~\cite{Spinrad03},
and for $K_{t}$-subdivision-free graphs in~\cite{LeeO14}.
	\qed
	\end{proof}

	In this subsection, we prove Theorem~\ref{the:poly-many-maxcliques}. 
	However, by Theorems~\ref{the:TS=TAR}(a) and \ref{the:TJ=TAR}(a) it suffices to give such an algorithm only for the $\TARrule$ rule.

	Let $(G, \cliq_{\ini}, \cliq_{\tar}, k)$ be any $\TARrule$-instance.
	Then, we define the \emph{$k$-intersection maximal-clique graph} of $G$, denoted by $\MC{k}{G}$, as follows:
		\begin{listing}{aaa}
		\item[(\one)] each node in $\MC{k}{G}$ corresponds to a clique in $\Mset{G}$; and 
		\item[(\two)] two nodes in $\MC{k}{G}$ are joined by an edge if and only if $|M \cap M^\prime| \ge k$ holds for the corresponding two maximal cliques $M$ and $M^\prime$ in $\Mset{G}$.  
		\end{listing}
	Note that any maximal clique in $\Mset{G}$ of size less than $k$ is contained in $\MC{k}{G}$ as an isolated node.
	We now give the key lemma to prove Theorem~\ref{the:poly-many-maxcliques}.
	\begin{lemma} \label{lem:clique-path}
	Let $G$ be a graph, and let $\cliq$ and $\cliq^\prime$ be any pair of cliques in $G$ such that $|\cliq| \ge k$ and $|\cliq^\prime| \ge k$.
	Let $M \supseteq \cliq$ and $M^\prime \supseteq \cliq^\prime$ be arbitrary maximal cliques in $\Mset{G}$.
	Then, $\cliq \sevstep \cliq^\prime$ under $\TAR{k}$ if and only if $\MC{k}{G}$ contains a path between the two nodes corresponding to $M$ and $M^\prime$.
	\end{lemma}

	\noindent
	{\bf Proof of Theorem~\ref{the:poly-many-maxcliques}.}

	For any graph $G$ with $n$ vertices and $m$ edges, Tsukiyama et al.~\cite{TsukiyamaIAS77} proved that the set $\Mset{G}$ can be computed in time $O \bigl(m n \lvert \Mset{G} \rvert \bigr)$.
	Thus, we can construct $\MC{k}{G}$ in time $O \bigl(m n \lvert \Mset{G} \rvert + n \lvert \Mset{G} \rvert^{2} \bigr)$.
	By the breadth-first search on $\MC{k}{G}$ which starts from an arbitrary maximal clique (node) $M \supseteq \cliq_{\ini}$, we can check in time $O \bigl(\lvert \Mset{G} \rvert^{2} \bigr)$ whether $\MC{k}{G}$ has a path to a maximal clique $M^\prime \supseteq \cliq_{\tar}$.
	Then, the theorem follows from Lemma~\ref{lem:clique-path}.
	\qed


	\section{Linear-Time Algorithm for Chordal Graphs}
	\label{sec:chordal}
	
	Since any chordal graph is even-hole free, by Corollary~\ref{cor:maximal} \textsc{clique reconfiguration} is solvable in polynomial time for chordal graphs. 
	Furthermore, we have discussed in Section~\ref{subsec:boundedclique} that the shortest variant is fixed-parameter tractable for chordal graphs when parameterized by the size of a maximum clique in a graph. 
	However, we give the following theorem in this section.
	\begin{theorem} \label{the:chordal}
	\textsc{Shortest clique reconfiguration} under any of $\TARrule$, $\TS$ and $\TJ$ can be solved in linear time for chordal graphs. 
	\end{theorem}

	In this section, we prove Theorem~\ref{the:chordal}. 
	By Theorems~\ref{the:TS=TAR}(a) and \ref{the:TJ=TAR}(a) it suffices to give a linear-time algorithm for a $\TARrule$-instance;
recall that the reduction from $\TS$/$\TJ$ to $\TARrule$ preserves the shortest length of reconfiguration sequences.

	Our algorithm consists of two phases.
	The first is a linear-time reduction from a given $\TARrule$-instance $(G, \cliq_{\ini}, \cliq_{\tar}, k)$ for a chordal graph $G$ to a $\TARrule$-instance $(\subH, \cliq_{\ini}, \cliq_{\tar}, k)$ for an interval graph $\subH$ such that $\distTARG{\subH}{\cliq_{\ini}}{\cliq_{\tar}}{k} = \distTARG{G}{\cliq_{\ini}}{\cliq_{\tar}}{k}$. 
	The second is a linear-time algorithm for interval graphs.

\medskip

\noindent
	{\bf Definitions of chordal graphs and interval graphs.}

	A graph is a \emph{chordal graph} if every induced cycle is of length three.
	Recall that $\Mset{G}$ is the set of all maximal cliques in a graph $G$, and we denote by $\Msetv{G}{v}$ the set of all maximal cliques in $G$ that contain a vertex $v \in V(G)$. 
A tree $\mathcal{T}$ is a  \emph{clique tree} of a graph $G$ if it satisfies the following conditions:
	\begin{listing}{aa}
	\item[-] each node in $\mathcal{T}$ corresponds to a maximal clique in $\mathcal{M}(G)$; and 
	\item[-] for each $v \in V(G)$, the subgraph of $\mathcal{T}$ induced by $\Msetv{G}{v}$ is connected.
	\end{listing}
	It is known that a graph is a chordal graph if and only if it has a clique tree~\cite{Gavril74}.
	A clique tree of a chordal graph can be computed in linear time (see~\cite[\S 15.1]{Spinrad03}).

	A graph is an \emph{interval graph} if it can be represented as the intersection graph of intervals on the real line.
	A \emph{clique path} is a clique tree which is a path.
	It is known that a graph is an interval graph if and only if it has a clique path~\cite{FulkersonG65,GilmoreH64}.
%



\subsection{Linear-time reduction from chordal graphs to interval graphs}

	In this subsection, we describe the first phase of our algorithm. 

	Let $(G, \cliq_{\ini}, \cliq_{\tar}, k)$ be any $\TARrule$-instance for a chordal graph $G$, and let $\mathcal{T}$ be a clique tree of $G$.
	Then, we find an arbitrary pair of maximal cliques $M_{\ini}$ and $M_{t}$ in $G$ (i.e., two nodes in $\mathcal{T}$) such that $\cliq_{\ini} \subseteq M_{\ini}$ and $\cliq_{\tar} \subseteq M_{t}$.
	Let $(M_{\ini}, M_1, \ldots, M_t)$ be the unique path in $\mathcal{T}$ from $M_{\ini}$ to $M_t$. 
	We define a graph $\subHp$ as the subgraph of $G$ induced by the maximal cliques $M_{\ini}, M_1, \ldots, M_t$.
	Note that $\subHp$ is an interval graph, because $(M_{\ini}, M_1, \ldots, M_t)$ forms a clique path. 
	

%

	The following lemma implies that the interval graph $\subHp$ has a $\TAR{k}$-sequence $\langle \cliq_{0}, \cliq_1, \ldots, \cliq_{\ell^\prime} \rangle$ such that $\ell^\prime = \distTARG{G}{\cliq_{\ini}}{\cliq_{\tar}}{k}$, and hence yields that $\distTARG{\subHp}{\cliq_{\ini}}{\cliq_{\tar}}{k} = \distTARG{G}{\cliq_{\ini}}{\cliq_{\tar}}{k}$ holds. 
	\begin{lemma} \label{lem:consecutive-cliques}
	Let $(G, \cliq_{\ini}, \cliq_{\tar}, k)$ be a $\TARrule$-instance for a chordal graph $G$, and let $\mathcal{T}$ be a clique tree of $G$.
	Suppose that $\langle \cliq_{\ini}, \cliq_1, \ldots, \cliq_{\ell} \rangle$ is a shortest $\TAR{k}$-sequence in $G$ from $\cliq_{\ini}$ to $\cliq_{\ell} = \cliq_{\tar}$.
	Let $(M_{0}, M_1, \dots M_{t})$ be the path in $\mathcal{T}$ from $M_{0}$ to $M_{t}$ for any pair of maximal cliques $M_{0} \supseteq \cliq_0$ and $M_{t} \supseteq \cliq_{\tar}$.
	Then, there is a monotonically increasing function $f \colon \{0,1,\dots,\ell\} \to \{0, 1,\dots, t\}$ such that $C_{i} \subseteq M_{f(i)}$ for each $i \in \{0,1,\ldots,\ell\}$.
\end{lemma}

	Although Lemma~\ref{lem:consecutive-cliques} implies that $\distTARG{\subHp}{\cliq_{\ini}}{\cliq_{\tar}}{k} = \distTARG{G}{\cliq_{\ini}}{\cliq_{\tar}}{k}$ holds for the interval graph $\subHp$, it seems difficult to find two maximal cliques $M_{0} \supseteq \cliq_0$ and $M_{t} \supseteq \cliq_{\tar}$ (and hence construct $\subHp$ from $G$) in linear time. 
	However, by a small trick, we can construct an interval graph $\subH$ in linear time such that $\distTARG{\subH}{\cliq_{\ini}}{\cliq_{\tar}}{k} = \distTARG{G}{\cliq_{\ini}}{\cliq_{\tar}}{k}$, as follows.
	\begin{lemma} \label{lem:chordal->interval}
	Given a $\TARrule$-instance $(G, \cliq_{\ini}, \cliq_{\tar}, k)$ for a chordal graph $G$, one can obtain a subgraph $\subH$ of $G$ in linear time such that $\subH$ is an interval graph, $\cliq_{\ini}, \cliq_{\tar} \subseteq V(\subH)$ and $\distTARG{\subH}{\cliq_{\ini}}{\cliq_{\tar}}{k} = \distTARG{G}{\cliq_{\ini}}{\cliq_{\tar}}{k}$.
	\end{lemma}

\subsection{Linear-time algorithm for interval graphs} \label{subsec:interval}

	In this subsection, we describe the second phase of our algorithm.

	Let $\intH$ be a given interval graph, and we assume that its clique path $\mathcal{P}$ has $V(\mathcal{P}) = \Mset{\intH} = \{M_{0}, M_{1}, \ldots, M_{t}\}$ and $E(\mathcal{P}) = \{\{M_{i}, M_{i+1}\} \mid 0 \le i < t\}$. 
	Note that we can assume that $t \ge 1$, that is, $\intH$ has at least two maximal cliques;
otherwise we can easily solve the problem in linear time (as in Lemma~\ref{lem:tar-dist-in-a-clique} in Appendix~\ref{app:oneclique}).
	For a vertex $v$ in $\intH$, let $l_{v} = \min\{i \mid v \in M_{i}\}$ and $r_{v} = \max\{i \mid v \in M_{i}\}$;
the indices $l_v$ and $r_v$ are called the \emph{$l$-value} and \emph{$r$-value} of $v$, respectively.
	Note that $v \in M_{i}$ if and only if $l_{v} \le i \le r_{v}$.
	For an interval graph $\intH$, such a clique path $\mathcal{P}$ and the indices $l_{v}$ and $r_{v}$ for all vertices $v \in V(\intH)$ can be computed in linear time~\cite{UeharaU07}.

	Let $(\intH, \cliq_{\ini}, \cliq_{\tar}, k)$ be a $\TARrule$-instance.
	We assume that $\cliq_{\ini} \subseteq M_{0}$, $\cliq_{\ini} \not\subseteq M_{1}$ and $\cliq_{\tar} \subseteq M_{t}$;
otherwise, we can remove the maximal cliques $M_{i}$ with $i < \min\{r_{v} \mid v \in \cliq_{\ini} \}$ and $i > \max\{l_{v} \mid v \in \cliq_{\tar} \}$ in linear time.
	Our algorithm greedily constructs a shortest $\TAR{k}$-sequence from $\cliq_{\ini}$ to $\cliq_{\tar}$, as follows:
	\begin{listing}{aaa}
	\item[(1)] if $\cliq_{\ini} \not\subseteq \cliq_{\tar}$ and $|\cliq_{\ini}| \ge k+1$, then remove a vertex with the minimum $r$-value in $\cliq_{\ini} \setminus \cliq_{\tar}$ from $\cliq_{\ini}$; 
	\item[(2)] otherwise add a vertex in $(\cliq_{\tar} \setminus \cliq_{\ini}) \cap M_{0}$ if any; 
					if no such vertex exists, add a vertex with the maximum $r$-value in $M_{0} \setminus \cliq_{\ini}$.
	\end{listing}
	We regard the clique obtained by the operations above as $\cliq_{\ini}$;
if necessary, we shift the indices of $M_i$ so that $\cliq_{\ini} \subseteq M_{0}$ and $\cliq_{\ini} \not\subseteq M_{1}$ hold; and repeat.
	If $\cliq_{\ini} \neq \cliq_{\tar}$ and none of the operations above is possible, we can conclude that $(\intH, \cliq_{\ini}, \cliq_{\tar}, k)$ is a $\NO$-instance.
	The correctness proof of this greedy algorithm and the estimation of its running time can be found in Appendix~\ref{app:interval-algo}.

	This completes the proof of Theorem~\ref{the:chordal}. 


	\section{Conclusion}
	In this paper, we have systematically shown that \textsc{clique reconfiguration} and its shortest variant can be solved in polynomial time for several graph classes. 
	As far as we know, this is the first example of a reconfiguration problem such that it is PSPACE-complete in general, but is solvable in polynomial time for such a variety of graph classes.

	\subsection*{Acknowledgments}
	This work is partially supported by MEXT/JSPS KAKENHI 25106504 and 25330003 (T.~Ito), 25104521, 26540005 and 26540005 (H.~Ono), and 24106004 and 25730003 (Y.~Otachi).

\bibliographystyle{abbrv}

\input{appendix.tex}


\end{document}

%% file: appendix.tex
\newpage
\appendix

	\section{Proofs Omitted from Section~\ref{sec:rules}}
	
	\subsection{Proof of Lemma~\ref{lem:TS->TAR}}

	To prove Lemma~\ref{lem:TS->TAR}, we first give the following lemma.	
	\begin{lemma} \label{lem:ar<=k+1}
	Let $G$ be a graph, and let $\cliq$ and $\cliq^\prime$ be any pair of cliques of $G$ such that $|\cliq| = |\cliq^\prime| = k$ and $\cliq \sevstep \cliq^\prime$ under $\TAR{k}$.
	Then, there exists a shortest $\TAR{k}$-sequence $\langle \cliq_{0}, \cliq_{1}, \dots, \cliq_{\ell} \rangle$ from $\cliq_{0} = \cliq$ to $\cliq_{\ell} = \cliq^\prime$ such that $|\cliq_{2i-1}| = k+1$ and $|\cliq_{2i}| = k$ for every $i \in \{1, 2, \ldots, \ell/2 \}$.
	\end{lemma}
	\begin{proof}
	Let $\langle \cliq_{0}, \cliq_{1}, \dots, \cliq_{\ell} \rangle$ be a shortest $\TAR{k}$-sequence from $\cliq_{0} = \cliq$ to $\cliq_{\ell} = \cliq^\prime$ which minimizes the sum $\sum_{i=0}^{\ell} |\cliq_{i}|$.
	Since each clique in the $\TAR{k}$-sequence $\langle \cliq_{0}, \cliq_{1}, \dots, \cliq_{\ell} \rangle$ is of size at least $k$, it suffices to show that $|C_{j}| \le k+1$ holds for every $j \in \{1,2,\ldots,\ell-1\}$.

	Let $s$ be an index satisfying $|\cliq_{s}| = \max_{i=0}^{\ell} |\cliq_{i}|$, and suppose for a contradiction that $|\cliq_{s}| \ge k + 2$.
	By the definition of $s$, we have $\cliq_{s-1} \subset \cliq_{s} \supset \cliq_{s+1}$.
	Let $\cliq_{s} = \cliq_{s-1} \cup \{a\}$ and $\cliq_{s+1} = (\cliq_{s-1} \cup \{a\}) \setminus \{b\}$.
	Note that, since $\langle \cliq_{0}, \cliq_{1}, \dots, \cliq_{\ell} \rangle$ is shortest, we have $a \neq b$ and hence $b \in \cliq_{s-1}$.
	We now replace the clique $\cliq_{s}$ by another clique $\cliq_{s}^\prime = \cliq_{s-1} \setminus \{b\}$, and obtain the following sequence $\mathcal{C}^\prime$ of cliques:
	\[
		\mathcal{C}^\prime = \langle \cliq_{0},\ \cliq_{1},\ \ldots,\ \cliq_{s-1},\ \cliq_{s-1} \setminus \{b\},\ \cliq_{s+1},\ \ldots,\ \cliq_{\ell} \rangle.
	\]
	Since $\cliq_{s-1} = \cliq_{s} \setminus \{a\}$ and $|\cliq_s| \ge k+2$, we have $|\cliq_{s}^\prime| = |\cliq_{s} \setminus \{a, b\}| \ge k$ and hence $\cliq_{s-1} \onestep \cliq_{s-1} \setminus \{b\} = \cliq_{s}^\prime$ under $\TAR{k}$.
	Furthermore, since $\cliq_{s+1} = (\cliq_{s-1} \cup \{a\}) \setminus \{b\} = \cliq_{s}^\prime \cup \{a\}$, we have $\cliq_{s}^\prime \onestep \cliq_{s+1}$ under $\TAR{k}$. 
	Therefore, $\mathcal{C}^\prime$ is a $\TAR{k}$-sequence between $\cliq$ and $\cliq^\prime$.

	Note that $\mathcal{C}^\prime$ is of length $\ell$, and hence it is a shortest $\TAR{k}$-sequence between $\cliq$ and $\cliq^\prime$.
	Since $\cliq_{s}^\prime = \cliq_{s} \setminus \{a, b\}$, we have $|\cliq_s^\prime| < |\cliq_s|$ and hence 
	\[
		|\cliq_s^\prime| + \sum \Bigl\{ |\cliq_{j}| : j \in \{0, 1, \ldots, \ell\} \setminus \{ s\}  \Bigr\} < \sum_{i=0}^{\ell} |\cliq_{i}|.
	\]
	This contradicts the assumption that $\langle \cliq_{0}, \cliq_{1}, \dots, \cliq_{\ell} \rangle$ is a shortest $\TAR{k}$-sequence from $\cliq_{0} = \cliq$ to $\cliq_{\ell} = \cliq^\prime$ which minimizes the sum $\sum_{i=0}^{\ell} |\cliq_{i}|$.
	\qed
	\end{proof}
\medskip

	\noindent
	{\bf Proof of Lemma~\ref{lem:TS->TAR}.}

	We first prove that $\TARins{\cliq_{\ini}}{\cliq_{\tar}}{k} = \YES$ if $\TSins{\cliq_{\ini}}{\cliq_{\tar}} = \YES$.
	In this case, there exists a $\TS$-sequence between $\cliq_{\ini}$ and $\cliq_{\tar}$; 
let $\langle \cliq_{0}, \cliq_{1}, \dots, \cliq_{\ell} \rangle$ be a shortest one, that is, $\cliq_{\ell} = \cliq_{\tar}$ and $\ell = \distTS{\cliq_{\ini}}{\cliq_{\tar}}$.
	Then, since this is a $\TS$-sequence, we have $u_{j-1} w_j \in E(G)$ for each $j \in \{1,2, \ldots, \ell\}$, where $\cliq_{j-1} \setminus \cliq_{j} = \{ u_{j-1} \}$ and $\cliq_{j} \setminus \cliq_{j-1} = \{ w_j\}$. 
(See \figurename~\ref{fig:TS}(a).)
	Therefore, $\cliq_{j-1} \cup \cliq_j$ $\bigl(= \cliq_{j-1} \cup \{ w_j \} \bigr)$ forms a clique of size $k+1$.
	Then, for each $j \in \{1,2, \ldots, \ell\}$, we replace each sub-sequence $\langle \cliq_j \rangle$ with $\langle \cliq_{j-1} \cup \{ w_{j} \},\ \cliq_j \rangle$, and obtain the following sequence $\mathcal{C}^\prime$ of cliques:
	\[
		\mathcal{C}^\prime = \langle \cliq_0,\ \cliq_0 \cup \{w_1 \},\ \cliq_1,\ \ldots, \cliq_{j-1} \cup \{ w_{j} \},\ \cliq_{j},\ \ldots,\cliq_{\ell-1} \cup \{ w_{\ell} \},\  \cliq_{\ell} \rangle.   
	\]
	Notice that $\cliq_{j-1} \cup \{ w_{j} \} \onestep \cliq_{j}$ under $\TAR{k}$ for each $j \in \{1,2, \ldots, \ell\}$, because $\bigl( \cliq_{j-1} \cup \{ w_{j} \} \bigr) \setminus \{ u_{j-1} \} = \cliq_j$.
	Therefore, the sequence $\mathcal{C}^\prime$ above is a $\TAR{k}$-sequence from $\cliq_0$ to $\cliq_{\ell} = \cliq_{\tar}$, and hence $\TARins{\cliq_{\ini}}{\cliq_{\tar}}{k} = \YES$. 
	Furthermore, by the construction, $\mathcal{C}^\prime$ is of length $2\ell$.
	Therefore, we have 
	\begin{equation} \label{eq:TAR<=2TS}
		\distTAR{\cliq_{\ini}}{\cliq_{\tar}}{k} \le 2 \ell = 2 \cdot \distTS{\cliq_{\ini}}{\cliq_{\tar}}.
	\end{equation}

	\begin{figure}[t]
		\centering
		\includegraphics[width=0.6\linewidth]{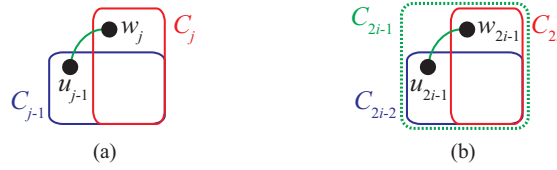}
	\vspace{-1em}
	\caption{Illustration for Lemma~\ref{lem:TS->TAR}.}
	\vspace{-1em}
	\label{fig:TS}
	\end{figure}

	We then prove that $\TSins{\cliq_{\ini}}{\cliq_{\tar}} = \YES$ if $\TARins{\cliq_{\ini}}{\cliq_{\tar}}{k} = \YES$.
	In this case, there exists a $\TAR{k}$-sequence between $\cliq_{\ini}$ and $\cliq_{\tar}$; 
let $\langle \cliq_{0}, \cliq_{1}, \dots, \cliq_{\ell^\prime} \rangle$ be a shortest one, that is, $\cliq_{\ell^\prime} = \cliq_{\tar}$ and $\ell^\prime = \distTAR{\cliq_{\ini}}{\cliq_{\tar}}{k}$.
	Furthermore, by Lemma~\ref{lem:ar<=k+1} we can assume that $|\cliq_{2i-1}| = k+1$ and $|\cliq_{2i}| = k$ for every $i \in \{1, 2, \ldots, \ell^\prime/2 \}$. 
	Then, observe that $\cliq_{2i-1} = \cliq_{2i-2} \cup \cliq_{2i}$ for every $i \in \{1, 2, \ldots, \ell^\prime/2 \}$, and let $\cliq_{2i-1} = \cliq_{2i-2} \cup \{w_{2i-1} \}  = \cliq_{2i} \cup \{u_{2i-1}\}$.
(See \figurename~\ref{fig:TS}(b).)
	Since this $\TAR{k}$-sequence $\langle \cliq_{0}, \cliq_{1}, \dots, \cliq_{\ell^\prime} \rangle$ is shortest, we have $u_{2i-1} \neq w_{2i-1}$.
	Furthermore, since both $u_{2i-1}$ and $w_{2i-1}$ belong to the clique $\cliq_{2i-1}$, they are adjacent.
	Therefore, for every $i \in \{1, 2, \ldots, \ell^\prime/2 \}$, we have $\cliq_{2i-2} \onestep C_{2i}$ under $\TS$;
we replace each sub-sequence $\langle \cliq_{2i-1}, \cliq_{2i} \rangle$ with $\langle \cliq_{2i} \rangle$, and obtain $\mathcal{C}^{\prime \prime} = \langle \cliq_0, \cliq_2, \cliq_4, \ldots, \cliq_{\ell^\prime} \rangle$.
	In this way, $\mathcal{C}^{\prime \prime}$ is a $\TS$-sequence from $\cliq_{\ini}$ to $\cliq_{\ell^\prime} = \cliq_{\tar}$, and hence $\TSins{\cliq_{\ini}}{\cliq_{\tar}} = \YES$. 
	Furthermore, the length of $\mathcal{C}^{\prime \prime}$ is $\ell^\prime/2$, and hence 
	\begin{equation} \label{eq:2TS<=TAR}
		\distTS{\cliq_{\ini}}{\cliq_{\tar}} \le \ell^\prime/2 = \distTAR{\cliq_{\ini}}{\cliq_{\tar}}{k}/2. 
	\end{equation}

	By Eqs.~(\ref{eq:TAR<=2TS}) and (\ref{eq:2TS<=TAR}) we have $\distTS{\cliq_{\ini}}{\cliq_{\tar}} = \distTAR{\cliq_{\ini}}{\cliq_{\tar}}{k} / 2$.
	\qed

	\subsection{Proof of Lemma~\ref{lem:TS=TARb}}
	Since $\cliq_{\ini}^\prime \subseteq \cliq_{\ini}$ and $|\cliq_{\ini}^\prime| = k$, we have $\cliq_{\ini} \sevstep \cliq_{\ini}^\prime$ under $\TAR{k}$ by deleting the vertices in $\cliq_{\ini} \setminus \cliq_{\ini}^\prime$ from $\cliq_{\ini}$ one by one. 
	Similarly, we have $\cliq_{\tar}^\prime \sevstep \cliq_{\tar}$ under $\TAR{k}$;
recall that any reconfiguration sequence is reversible.
	Since $|\cliq_{\ini}^\prime| = |\cliq_{\tar}^\prime| = k$, by Lemma~\ref{lem:TS->TAR} we have 
	\begin{equation} \label{eq:TS=TAR}
		\TSins{\cliq_{\ini}^\prime}{\cliq_{\tar}^\prime} = \TARins{\cliq_{\ini}^\prime}{\cliq_{\tar}^\prime}{k}.
	\end{equation}

	We now prove that $\TARins{\cliq_{\ini}}{\cliq_{\tar}}{k} = \YES$ if $\TSins{\cliq_{\ini}^\prime}{\cliq_{\tar}^\prime} = \YES$.
	In this case, by Eq.~(\ref{eq:TS=TAR}) we have $\TARins{\cliq_{\ini}^\prime}{\cliq_{\tar}^\prime}{k} = \YES$ and hence $\cliq_{\ini}^\prime \sevstep \cliq_{\tar}^\prime$ under $\TAR{k}$.
	Thus, $\cliq_{\ini} \sevstep \cliq_{\ini}^\prime \sevstep \cliq_{\tar}^\prime \sevstep \cliq_{\tar}$ holds under $\TAR{k}$, and hence $\TARins{\cliq_{\ini}}{\cliq_{\tar}}{k} = \YES$. 
	
	We finally prove that $\TSins{\cliq_{\ini}^\prime}{\cliq_{\tar}^\prime} = \YES$ if $\TARins{\cliq_{\ini}}{\cliq_{\tar}}{k} = \YES$.
	In this case, since $\TARins{\cliq_{\ini}}{\cliq_{\tar}}{k} = \YES$, we have $\cliq_{\ini} \sevstep \cliq_{\tar}$ under $\TAR{k}$.
	Therefore, $\cliq_{\ini}^\prime \sevstep \cliq_{\ini} \sevstep \cliq_{\tar} \sevstep \cliq_{\tar}^\prime$ holds under $\TAR{k}$, and hence $\TARins{\cliq_{\ini}^\prime}{\cliq_{\tar}^\prime}{k} = \YES$. 
	By Eq.~(\ref{eq:TS=TAR}) we then have $\TSins{\cliq_{\ini}^\prime}{\cliq_{\tar}^\prime} = \YES$.
\qed

	\subsection{Proof of Lemma~\ref{lem:TJ->TAR}}
	
	We first give the following lemma, which can be obtained from the same arguments as in Lemma~\ref{lem:ar<=k+1} by just shifting the threshold by one.	
	\begin{lemma} \label{lem:ar<=k}
	Let $G$ be a graph, and let $\cliq$ and $\cliq^\prime$ be any pair of cliques of $G$ such that $|\cliq| = |\cliq^\prime| = k$ and $\cliq \sevstep \cliq^\prime$ under $\TAR{k-1}$.
	Then, there exists a shortest $\TAR{k-1}$-sequence $\langle \cliq_{0}, \cliq_{1}, \dots, \cliq_{\ell} \rangle$ from $\cliq_{0} = \cliq$ to $\cliq_{\ell} = \cliq^\prime$ such that $|\cliq_{2i-1}| = k-1$ and $|\cliq_{2i}| = k$ for every $i \in \{1, 2, \ldots, \ell/2 \}$.
	\end{lemma}

	\noindent
	{\bf Proof of Lemma~\ref{lem:TJ->TAR}.}

	We first prove that $\TARins{\cliq_{\ini}}{\cliq_{\tar}}{k-1} = \YES$ if $\TJins{\cliq_{\ini}}{\cliq_{\tar}} = \YES$.
	In this case, there exists a $\TJ$-sequence between $\cliq_{\ini}$ and $\cliq_{\tar}$; 
let $\langle \cliq_{0}, \cliq_{1}, \dots, \cliq_{\ell} \rangle$ be a shortest one, that is, $\cliq_{\ell} = \cliq_{\tar}$ and $\ell = \distTJ{\cliq_{\ini}}{\cliq_{\tar}}$.
	For each $j \in \{1,2, \ldots, \ell\}$, let $\cliq_{j-1} \setminus \cliq_{j} = \{ u_{j-1} \}$ and $\cliq_{j} \setminus \cliq_{j-1} = \{ w_j\}$. 
	Then, we replace each sub-sequence $\langle \cliq_j \rangle$ with $\langle \cliq_{j-1} \setminus \{ u_{j-1} \},\ \cliq_j \rangle$ for each $j \in \{1,2, \ldots, \ell\}$, and obtain the following sequence $\mathcal{C}^\prime$ of cliques:
	\[
		\mathcal{C}^\prime = \langle \cliq_0,\ \cliq_0 \setminus \{u_0 \},\ \cliq_1,\ \ldots, \cliq_{j-1} \setminus \{ u_{j-1} \},\ \cliq_{j},\ \ldots,\ \cliq_{\ell-1} \setminus \{ u_{\ell-1} \},\ \cliq_{\ell} \rangle.   
	\]
	Notice that $\cliq_{j-1} \setminus \{ u_{j-1} \} \onestep \cliq_{j}$ under $\TAR{k-1}$ for each $j \in \{1,2, \ldots, \ell\}$, because $\bigl( \cliq_{j-1} \setminus \{ u_{j-1} \} \bigr) \cup \{ w_j\} = \cliq_j$ and $\bigl|\cliq_{j-1} \setminus \{ u_{j-1} \} \bigr| = k-1$.
	Therefore, the sequence $\mathcal{C}^\prime$ above is a $\TAR{k-1}$-sequence from $\cliq_{\ini}$ to $\cliq_{\ell} = \cliq_{\tar}$, and hence $\TARins{\cliq_{\ini}}{\cliq_{\tar}}{k-1} = \YES$. 
	Furthermore, by the construction, $\mathcal{C}^\prime$ is of length $2\ell$.
	Therefore, we have 
	\begin{equation} \label{eq:TAR<=2TJ}
		\distTAR{\cliq_{\ini}}{\cliq_{\tar}}{k-1} \le 2 \ell = 2 \cdot \distTJ{\cliq_{\ini}}{\cliq_{\tar}}.
	\end{equation}
	
	We then prove that $\TJins{\cliq_{\ini}}{\cliq_{\tar}} = \YES$ if $\TARins{\cliq_{\ini}}{\cliq_{\tar}}{k-1} = \YES$.
	In this case, there exists a $\TAR{k-1}$-sequence between $\cliq_{\ini}$ and $\cliq_{\tar}$; 
let $\langle \cliq_{0}, \cliq_{1}, \dots, \cliq_{\ell^\prime} \rangle$ be a shortest one, that is, $\cliq_{\ell^\prime} = \cliq_{\tar}$ and $\ell^\prime = \distTAR{\cliq_{\ini}}{\cliq_{\tar}}{k-1}$.
	Furthermore, by Lemma~\ref{lem:ar<=k} we can assume that $|\cliq_{2i-1}| = k-1$ and $|\cliq_{2i}| = k$ for every $i \in \{1, 2, \ldots, \ell^\prime/2 \}$. 
	For every $i \in \{1, 2, \ldots, \ell^\prime/2 \}$, let $\cliq_{2i-1} = \cliq_{2i-2} \setminus \{u_{2i-2} \}$ and $\cliq_{2i}  = \cliq_{2i-1} \cup \{w_{2i-1}\}$.
	Since $\langle \cliq_{0}, \cliq_{1}, \dots, \cliq_{\ell^\prime} \rangle$ is shortest, we have $u_{2i-2} \neq w_{2i-1}$.
	Then, for every $i \in \{1, 2, \ldots, \ell^\prime/2 \}$, we have $\cliq_{2i-2} \onestep C_{2i}$ under $\TJ$;
we replace each sub-sequence $\langle \cliq_{2i-1}, \cliq_{2i} \rangle$ with $\langle \cliq_{2i} \rangle$, and obtain $\mathcal{C}^{\prime \prime} = \langle \cliq_0, \cliq_2, \cliq_4, \ldots, \cliq_{\ell^\prime} \rangle$.
	In this way, $\mathcal{C}^{\prime \prime}$ is a $\TJ$-sequence from $\cliq_{\ini}$ to $\cliq_{\ell^\prime} = \cliq_{\tar}$, and hence $\TJins{\cliq_{\ini}}{\cliq_{\tar}} = \YES$. 
	Furthermore, the length of $\mathcal{C}^{\prime \prime}$ is $\ell^\prime/2$, and hence 
	\begin{equation} \label{eq:2TJ<=TAR}
		\distTJ{\cliq_{\ini}}{\cliq_{\tar}} \le \ell^\prime/2 = \distTAR{\cliq_{\ini}}{\cliq_{\tar}}{k-1}/2. 
	\end{equation}

	By Eqs.~(\ref{eq:TAR<=2TJ}) and (\ref{eq:2TJ<=TAR}) we have $\distTJ{\cliq_{\ini}}{\cliq_{\tar}} = \distTAR{\cliq_{\ini}}{\cliq_{\tar}}{k-1} /2$.
	\qed

	\subsection{Proof of Lemma~\ref{lem:TJ=TARb}}
	Similarly as in the proof of Lemma~\ref{lem:TS=TARb}, in both cases (\one) and (\two), we have $\cliq_{\ini} \sevstep \cliq_{\ini}^\prime$ and $\cliq_{\tar} \sevstep \cliq_{\tar}^\prime$ under $\TAR{k}$.
	Note that $|\cliq_{\ini}^\prime| = |\cliq_{\tar}^\prime| = k+1$. 
	Then, by Lemma~\ref{lem:TJ->TAR} we have
	\begin{equation} \label{eq:TJ=TAR}
		\TJins{\cliq_{\ini}^\prime}{\cliq_{\tar}^\prime} = \TARins{\cliq_{\ini}^\prime}{\cliq_{\tar}^\prime}{k}.
	\end{equation}
	
	We first prove that $\TARins{\cliq_{\ini}}{\cliq_{\tar}}{k} = \YES$ if $\TJins{\cliq_{\ini}^\prime}{\cliq_{\tar}^\prime} = \YES$.
	In this case, by Eq.~(\ref{eq:TJ=TAR}) we have $\TARins{\cliq_{\ini}^\prime}{\cliq_{\tar}^\prime}{k} = \YES$, and hence $\cliq_{\ini}^\prime \sevstep \cliq_{\tar}^\prime$ under $\TAR{k}$.
	Thus, $\cliq_{\ini} \sevstep \cliq_{\ini}^\prime \sevstep \cliq_{\tar}^\prime \sevstep \cliq_{\tar}$ holds under $\TAR{k}$, and hence $\TARins{\cliq_{\ini}}{\cliq_{\tar}}{k} = \YES$. 
	
	We then prove that $\TJins{\cliq_{\ini}^\prime}{\cliq_{\tar}^\prime} = \YES$ if $\TARins{\cliq_{\ini}}{\cliq_{\tar}}{k} = \YES$.
	In this case, since $\TARins{\cliq_{\ini}}{\cliq_{\tar}}{k} = \YES$, we have $\cliq_{\ini} \sevstep \cliq_{\tar}$ under $\TAR{k}$.
	Therefore, $\cliq_{\ini}^\prime \sevstep \cliq_{\ini} \sevstep \cliq_{\tar} \sevstep \cliq_{\tar}^\prime$ holds under $\TAR{k}$, and hence $\TARins{\cliq_{\ini}^\prime}{\cliq_{\tar}^\prime}{k} = \YES$. 
	By Eq.~(\ref{eq:TJ=TAR}) we then have $\TJins{\cliq_{\ini}^\prime}{\cliq_{\tar}^\prime} = \YES$.
\qed

	\subsection{Proof of Proposition~\ref{pro:perfect}}
	Kami\'nski et al.~\cite[Theorem~3]{KaminskiMM12} proved that \textsc{independent set reconfiguration} under $\TARrule$ is PSPACE-complete for perfect graphs.
	Since the class of perfect graphs is closed under taking complements~\cite{Lovasz72}, by Lemma~\ref{lem:clique-independent} \textsc{clique reconfiguration} under $\TARrule$ is PSPACE-complete for perfect graphs.
	Then, Theorems~\ref{the:TS=TAR}(b) and \ref{the:TJ=TAR}(b) imply that \textsc{clique reconfiguration} remains PSPACE-complete for perfect graphs under $\TS$ and $\TJ$, too.
	\qed

	\subsection{Proof of Proposition~\ref{pro:perfect}}
	From the definition, the class of cographs is closed under taking complements, and we note that the complement of a cograph can be computed in linear time~\cite{CorneilPS85}.
	Bonsma~\cite{Bon14} proved that \textsc{independent set reconfiguration} under $\TARrule$ is solvable in linear time for cographs, and hence by Lemma~\ref{lem:clique-independent} we can solve \textsc{clique reconfiguration} under $\TARrule$ in linear time for cographs. 
	Then, Theorems~\ref{the:TS=TAR}(a) and \ref{the:TJ=TAR}(a) imply that \textsc{clique reconfiguration} can be solved in linear time for cographs under $\TS$ and $\TJ$, too.
	\qed

	\section{Proofs Omitted from Section~\ref{sec:polytime}}
	
	\subsection{Proof of Theorem~\ref{the:boundedclique}}
	By Theorems~\ref{the:TS=TAR}(a) and \ref{the:TJ=TAR}(a) it suffices to give an $O(w^{2} n^{w})$-time algorithm for a $\TARrule$-instance;
	recall that the reduction from $\TS$/$\TJ$ to $\TARrule$ preserves the shortest length of reconfiguration sequences.
	Note that, however, the arguments for $\TARrule$ below can be applied to the other rules $\TS$ and $\TJ$, and one can obtain algorithms directly for $\TS$ and $\TJ$ rules.

	Let $(G, \cliq_{\ini}, \cliq_{\tar}, k)$ be any $\TARrule$-instance such that $\omega(G) \le w$.
	Then, the number of cliques of size at least $k$ in $G$ can be bounded by $\sum_{i=k}^{w} {n \choose i} = O(n^w)$.
	We now construct a \textit{reconfiguration graph} $\cgraph = (\cvertex, \cedge)$, as follows: 
		\begin{listing}{aaa}
		\item[(\one)] each node in $\cgraph$ corresponds to a clique of $G$ with size at least $k$; and 
		\item[(\two)] two nodes in $\cgraph$ are joined by an edge if and only if $\cliq \onestep \cliq^\prime$ holds under $\TAR{k}$ for the corresponding two cliques $\cliq$ and $\cliq^\prime$.  
		\end{listing}
	This reconfiguration graph $\cgraph$ can be constructed in time $O(w^{2} n^{w})$ as follows:
	we first enumerate all cliques in time $O(w^{2} n^{w})$ by checking all $O(n^{w})$ vertex subsets of size at most $w$;
	we then add edges from each clique to its $O(w)$ subsets with one less vertex.
	The graph $\cgraph$ has $|\cvertex| = O(n^w)$ nodes and $|\cedge| = O(w n^w)$ edges. 
	Then, there is a $\TAR{k}$-sequence between $\cliq_{\ini}$ and $\cliq_{\tar}$ if and only if there is a path in $\cgraph$ between the two corresponding nodes. 
	Therefore, by the breadth-first search on $\cgraph$ which starts from the node corresponding to $\cliq_{\ini}$, we can check if $\cgraph$ has a desired path or not in time $O(|\cvertex| + |\cedge|) = O(w n^{w})$. 
	Furthermore, if such a path exists, it corresponds to a shortest $\TAR{k}$-sequence between $\cliq_{\ini}$ and $\cliq_{\tar}$. 
	\qed

	\subsection{Proof of Proposition~\ref{pro:treewidth}}
	We first compute a tree-decomposition $\mathcal{T}$ with width $5t + 4$ in $O(c^{t} n)$ time, where $c$ is some constant,
	by using the algorithm in \cite{BodlaenderDDFLP13}.
	Additionally, we can assume that the number of bags in $\mathcal{T}$ is $O(n)$~\cite{BodlaenderDDFLP13}.
	By the definition of the tree-decomposition, every clique in $G$ is included in at least one bag of $\mathcal{T}$.
	Since the width of $\mathcal{T}$ is $5t + 4$, each bag in $\mathcal{T}$ contains at most $5t+5$ vertices of $G$.
	Thus, there are at most $2^{5t+5}$ cliques in each bag of $\mathcal{T}$, and hence we can conclude that $G$ has $O(2^{5t+5} n)$ cliques.
	Then, the proposition follows, because we can construct a reconfiguration graph $\cgraph$ in time $O(t^2 2^{5t+5} n)$, similarly as in the proof of Theorem~\ref{the:boundedclique}.
	\qed

	\subsection{Proof of Lemma~\ref{lem:clique-path}}
	We first prove the if-part.
	Suppose that there is a path $\langle M_{0}, M_{1}, \dots, M_{\ell} \rangle$ in $\MC{k}{G}$ from the node $M = M_{0} \supseteq \cliq$ to the node $M^\prime = M_{\ell} \supseteq \cliq^\prime$.
	Let $\cliq_0 = \cliq$, and let $\cliq_j$ be any clique in $M_{j-1} \cap M_{j}$ of size $k$ for each $j \in \{1, 2, \ldots, \ell\}$; 
such a clique $\cliq_j$ exists because $|M_{j-1} \cap M_j| \ge k$. 
	Then, $\cliq_{j-1} \sevstep \cliq_j$ holds under $\TAR{k}$ because $\cliq_{j-1} \cup \cliq_j \subseteq M_{j-1}$ and hence $\cliq_{j-1} \cup \cliq_j$ forms a clique of $G$ for each $j \in \{1, 2, \ldots, \ell\}$.
	We thus have $\cliq = \cliq_0 \sevstep \cliq_1 \sevstep \cdots \sevstep \cliq_\ell$ under $\TAR{k}$. 
	Since both $\cliq_{\ell}$ and $\cliq^\prime$ are contained in the same maximal clique $M_{\ell} = M^\prime$, we have $\cliq_{\ell} \sevstep \cliq^\prime$ and hence $\cliq \sevstep \cliq^\prime$ holds under $\TAR{k}$.
 
	We then prove the only-if-part.
	Suppose that there is a $\TAR{k}$-sequence $\mathcal{C} = \langle \cliq_{0}, \cliq_{1}, \ldots, \cliq_{\ell^\prime} \rangle$ such that $\cliq_{0} = \cliq$ and $\cliq_{\ell^\prime} = \cliq^\prime$.
	Let $\MC{k}{G; \mathcal{C}}$ be the subgraph of $\MC{k}{G}$ induced by all nodes (i.e., maximal cliques in $\Mset{G}$) that contain at least one clique in $\mathcal{C}$.
	Then, it suffices to show that $\MC{k}{G; \mathcal{C}}$ is connected;
then $\MC{k}{G}$ has a path from any node $M \supseteq \cliq$ to any node $M^\prime \supseteq \cliq^\prime$. 
	Suppose for a contradiction that $\MC{k}{G; \mathcal{C}}$ is not connected. 
	Then, there exists an index $j$ such that the cliques $\cliq_{j-1}$ and $\cliq_j$ are contained in different maximal cliques $M_{p-1}$ and $M_p$ which belong to different connected components in $\MC{k}{G; \mathcal{C}}$. 
	In this case, $\cliq_{j}$ must be obtained by adding a vertex $u$ to $\cliq_{j-1}$, that is, $\cliq_{j} = \cliq_{j-1} \cup \{ u\}$;
otherwise both $\cliq_{j-1}$ and $\cliq_j$ are contained in the same maximal clique $M_{p-1}$.
	Since $\mathcal{C}$ is a $\TAR{k}$-sequence, we have $|\cliq_{j-1}| \ge k$ and hence $|\cliq_{j-1} \cap \cliq_j| \ge k$.
	Then, since $\cliq_{j-1} \subseteq M_{p-1}$ and $\cliq_j \subseteq M_p$, we have $|M_{p-1} \cap M_p| \ge k$. 
	Therefore, $M_{p-1}$ and $M_p$ must be joined by an edge in $\MC{k}{G}$ and hence in $\MC{k}{G; \mathcal{C}}$.
	This contradicts the assumption that $M_{p-1}$ and $M_{p}$ are contained in different connected components in $\MC{k}{G; \mathcal{C}}$. 
	We have thus proved that $\MC{k}{G; \mathcal{C}}$ is connected, and hence there is a path in $\MC{k}{G}$ from any node $M \supseteq \cliq$ to any node $M^\prime \supseteq \cliq^\prime$.
	\qed

	\section{Proofs Omitted from Section~\ref{sec:chordal}}
	
	\subsection{Proof of Lemma~\ref{lem:consecutive-cliques}} \label{app:oneclique}
	We first prove the following lemma, which can be applied to any graph. 
	\begin{lemma} \label{lem:tar-dist-in-a-clique}
	For two cliques $\cliq$ and $\cliq^\prime$ in a graph $G$, suppose that $\cliq \cup \cliq^\prime$ also forms a clique in $G$.
	Then, $\distTAR{\cliq}{\cliq^\prime}{k} = |\symdiff{\cliq}{\cliq^\prime}|$ for every integer $k \ge \min\{|\cliq|, |\cliq^\prime|\}$.
	Furthermore,  every clique in an arbitrary shortest $\TAR{k}$-sequence from $\cliq$ to $\cliq^\prime$ consists only of vertices in $\cliq \cup \cliq^\prime$.
	\end{lemma}
	\begin{proof}
	We first prove that $\distTAR{\cliq}{\cliq^\prime}{k} \le |\symdiff{\cliq}{\cliq^\prime}|$ holds for every integer $k \ge \min\{|\cliq|, |\cliq^\prime|\}$, by constructing a $\TAR{k}$-sequence between $\cliq$ and $\cliq^\prime$ of length $|\symdiff{\cliq}{\cliq^\prime}|$, as follows:
we first add the vertices in $\cliq^\prime \setminus \cliq$ to $\cliq$ one by one, and obtain the clique $\cliq \cup \cliq^\prime$; and 
we then delete the vertices in $\cliq \setminus \cliq^\prime$ from $\cliq \cup \cliq^\prime$ one by one, and obtain the clique $\cliq^\prime$.
	Since the minimum size of a clique in this sequence is $\min \{ |\cliq|, |\cliq^\prime| \}$, this is a $\TAR{k}$-sequence for every integer $k \ge \min\{|\cliq|, |\cliq^\prime|\}$. 
	Furthermore, the length of this $\TAR{k}$-sequence is $|\symdiff{\cliq}{\cliq^\prime}|$.
	Therefore, we have $\distTAR{\cliq}{\cliq^\prime}{k} \le |\symdiff{\cliq}{\cliq^\prime}|$.
	
	We then prove that $\distTAR{\cliq}{\cliq^\prime}{k} \ge |\symdiff{\cliq}{\cliq^\prime}|$ holds for every integer $k \ge \min\{|\cliq|, |\cliq^\prime|\}$.
	Since $k \ge \min\{|\cliq|, |\cliq^\prime|\}$, there exists at least one $\TAR{k}$-sequence between $\cliq$ and $\cliq^\prime$ as explained above. 
	Note that, in an arbitrary $\TAR{k}$-sequence between $\cliq$ and $\cliq^\prime$, every vertex in $\symdiff{\cliq}{\cliq^\prime}$ must be either deleted or added at least once. 
	Therefore, we have $\distTAR{\cliq}{\cliq^\prime}{k} \ge |\symdiff{\cliq}{\cliq^\prime}|$.
	
	We have thus proved that $\distTAR{\cliq}{\cliq^\prime}{k} = |\symdiff{\cliq}{\cliq^\prime}|$ holds for every integer $k \ge \min\{|\cliq|, |\cliq^\prime|\}$.
	Consider an arbitrary shortest $\TAR{k}$-sequence $\mathcal{C}$ from $\cliq$ to $\cliq^\prime$. 
	Then, every vertex in $\symdiff{\cliq}{\cliq^\prime}$ must be either deleted or added by $\mathcal{C}$ at least once.
	Therefore, if $\mathcal{C}$ deletes or adds a vertex not in $\cliq \cup \cliq^\prime$, then the length of $\mathcal{C}$ is strictly greater than $|\symdiff{\cliq}{\cliq^\prime}|$.
	This contradicts the assumption that $\mathcal{C}$ is shortest.
	We can thus conclude that every clique in an arbitrary shortest $\TAR{k}$-sequence from $\cliq$ to $\cliq^\prime$ consists only of vertices in $\cliq \cup \cliq^\prime$.
	\qed
	\end{proof}

	Let $G = (V,E)$ be a graph, and let $X, Y \subseteq V$.
	A vertex subset $S \subseteq V$ is called an \emph{$(X, Y)$-separator} of $G$ if any two vertices $x \in X \setminus S$ and $y \in Y \setminus S$ do not belong to the same component in $G-S$, where $G-S$ denotes the subgraph of $G$ induced by the vertex set $V \setminus S$.
\medskip

	\noindent
	{\bf Proof of Lemma~\ref{lem:consecutive-cliques}.}

	We prove the statement by induction on the length $t$ of the unique path $(M_0, M_1, \ldots, M_t)$ in $\mathcal{T}$ between $M_{0}$ and $M_{t}$.

	First, consider the case where $t = 0$.
	Then, since $\cliq_{\ini} \subseteq M_0$ and $\cliq_{\tar} \subseteq M_t = M_0$, both $\cliq_{\ini}$ and $\cliq_{\tar}$ are contained in the same maximal clique $M_{0}$.
	Therefore, $\cliq_{\ini} \cup \cliq_{\tar}$ forms a clique, and hence by Lemma~\ref{lem:tar-dist-in-a-clique} every shortest $\TAR{k}$-sequence passes through cliques consisting of vertices only in $M_{0}$.
	Thus, we set $f(i) = 0$ for all $i \in \{0, 1, \ldots, \ell\}$.
\medskip

	Next, consider the case where $t \ge 1$. 
	We assume that $\cliq_{\tar} \not\subseteq M_0$, because otherwise we can set $f(i) = 0$ for all $i \in \{0, 1, \ldots, \ell\}$ similarly as for the case $t=0$.
	Then, by the definition of a clique tree, $M_{0} \cap M_{1}$ forms a $(\cliq_{\ini}, \cliq_{\tar})$-separator of $G$ (see~\cite[Lemma 4.2]{BP93chordal}).

	We now claim that there exists at least one clique $\cliq_j$ in the shortest $\TAR{k}$-sequence $\langle \cliq_{\ini}, \cliq_1, \ldots, \cliq_{\ell} \rangle$ such that $\cliq_{j} \subseteq M_{0} \cap M_{1}$.
	Suppose for a contradiction that $\cliq_{i} \not\subseteq M_{0} \cap M_{1}$ for all $i \in \{0, 1, \ldots, \ell\}$.
	Let $w_{i}$ be an arbitrary vertex in $\cliq_{i} \setminus (M_{0} \cap M_{1})$ for each $i \in \{0, 1, \ldots, \ell\}$.
	Since $\langle \cliq_{\ini}, \cliq_1, \ldots, \cliq_{\ell} \rangle$ is a $\TAR{k}$-sequence, either $\cliq_{i} \subset \cliq_{i+1}$ or $\cliq_{i} \supset \cliq_{i+1}$ holds for each $i \in \{0, 1, \ldots, \ell-1\}$ and hence $\cliq_{i} \cup \cliq_{i+1}$ forms a clique.
	Therefore, the vertices $w_{i}$ and $w_{i+1}$ in $\cliq_{i} \cup \cliq_{i+1}$ are either the same or adjacent.
	This implies that the subgraph of $G$ induced by $\{ w_{i} \mid 0 \le i \le \ell\}$ is connected, and hence it contains a path from $w_{0}$ to $w_{\ell}$.
	However, since $w_0 \in \cliq_{0} \setminus (M_{0} \cap M_{1})$ and $w_{\ell} \in \cliq_{\ell} \setminus (M_{0} \cap M_{1}) = \cliq_{\tar} \setminus (M_{0} \cap M_{1})$, this contradicts the assumption that $M_{0} \cap M_{1}$ is a $(\cliq_{0}, \cliq_{\tar})$-separator.
 
	As the induction hypothesis, assume that the statement is true for the length $t-1 \ge 0$.
	Let $\cliq_j$ be an arbitrary clique in $\langle \cliq_{\ini}, \cliq_1, \ldots, \cliq_{\ell} \rangle$ such that $\cliq_{j} \subseteq M_{0} \cap M_{1}$. 
	Note that, since $\langle \cliq_{\ini}, \cliq_1, \ldots, \cliq_{\ell} \rangle$ is shortest, $\langle \cliq_{\ini}, \cliq_1, \ldots, \cliq_{j} \rangle$ is a shortest $\TAR{k}$-sequence from $\cliq_{\ini}$ to $\cliq_j$. 
	Then, since $\cliq_{\ini} \cup \cliq_{j} \subseteq M_{0}$, Lemma~\ref{lem:tar-dist-in-a-clique} implies that $\langle \cliq_{\ini}, \cliq_1, \ldots, \cliq_{j} \rangle$ passes through cliques consisting of vertices only in $M_{0}$, that is,
	\begin{equation} \label{eq:mono-one}
		\cliq_{h} \subseteq M_{0}
	\end{equation}
holds for each $h \in \{0, 1, \ldots, j\}$.
	Let $\cliq_{i}^\prime = \cliq_{j+i}$ for each $i \in \{0, 1, \ldots, \ell - j\}$, and let $M_{i}^\prime = M_{1+i}$ for each $i \in \{0,1,\ldots, t-1\}$.
	Note that $\langle \cliq_{\ini}^\prime, \cliq_1^\prime, \ldots, \cliq_{\ell-j}^\prime \rangle$ is a shortest $\TAR{k}$-sequence from $\cliq_0^\prime = \cliq_j$ to $\cliq_{\ell-j}^\prime = \cliq_{\ell} = \cliq_{\tar}$. 
	Furthermore, $\cliq_0^\prime = \cliq_j \subseteq M_1 = M_0^\prime$, $\cliq_{\ell-j}^\prime = \cliq_{\tar} \subseteq M_t = M_{t-1}^\prime$ and  $(M_0^\prime, M_1^\prime, \ldots, M_{t-1}^\prime)$ is a path in $\mathcal{T}$ of length $t-1$.
	Therefore, by the induction hypothesis, there is a monotonically increasing function $f^\prime \colon \{0, 1, \dots, \ell-j\} \to \{0, 1, \dots, t-1\}$ such that 
	\begin{equation} \label{eq:mono-two}
		\cliq_{i}^\prime \subseteq M_{f^\prime(i)}^\prime
	\end{equation}
for all $i \in \{0, 1, \dots, \ell-j\}$.
	Now we construct a mapping $f \colon \{0, 1, \dots, \ell\} \to \{0, 1, \dots, t\}$, as follows:
	\[
		f(i) =
			\begin{cases}
			0 & \text{if} \ i < j, \\    
			f^\prime(i - j) + 1 & \text{otherwise}.
			\end{cases}
	\]
	Since $f^\prime$ is a monotonically increasing function, $f$ is too. 
	Furthermore, by Eqs.~(\ref{eq:mono-one}) and (\ref{eq:mono-two}) we have $\cliq_{i} \subseteq M_{f(i)}$ for all $i \in \{0, 1, \dots, \ell\}$.
	Thus, $f$ satisfies the desired property.
	\qed

	\subsection{Proof of Lemma~\ref{lem:chordal->interval}}
	Before giving our linear-time reduction, we give the following lemma.

	\begin{lemma} \label{lem:never-appear-again}
	Suppose that $\langle \cliq_{0}, \cliq_1, \ldots, \cliq_{\ell}\rangle$ is a shortest $\TAR{k}$-sequence in a chordal graph $G$.
	Let $p$ and $q$ be two indices in $\{0,1, \ldots, \ell\}$ such that $p < q$.
	If there is a vertex $v$ in $\cliq_{p} \cap \cliq_{q}$, then $v \in \cliq_{i}$ holds for all $i \in \{p, p+1, \ldots, q\}$.
	\end{lemma}
	\begin{proof}
	Suppose for a contradiction that the statement does not hold.
	We may assume without loss of generality that $v \notin \cliq_{i}$ for every $i \in \{p+1, p+2, \ldots, q-1\}$ by setting $p$ as large as possible and $q$ as small as possible.
	Then, observe that $\cliq_{p+1} \cup \{v\} = \cliq_{p}$ and $\cliq_{q-1} \cup \{v\} = \cliq_{q}$.

	Let $\mathcal{T}$ be a clique tree of $G$.
	Let $(M_{0}, M_1, \dots M_{t})$ be the path in $\mathcal{T}$ from $M_{0}$ to $M_{t}$ for any pair of maximal cliques $M_{0} \supseteq \cliq_0$ and $M_{t} \supseteq \cliq_{\ell}$.
	By Lemma~\ref{lem:consecutive-cliques} there is a monotonically increasing function $f \colon \{0,1,\dots,\ell\} \to \{0, 1,\dots, t\}$ such that $C_{i} \subseteq M_{f(i)}$ for each $i \in \{0,1,\ldots,\ell\}$.
	Then, $f(p) \le f(i) \le f(q)$ for each $i \in \{p+1, p+2, \ldots, q-1\}$.
	Recall that, by the definition of a clique tree, the subgraph of $\mathcal{T}$ induced by $\Msetv{G}{v}$ is connected.
	Since $v \in \cliq_{p} \cap \cliq_{q} \subseteq M_{f(p)} \cap M_{f(q)}$, we can conclude that the vertex $v$ is contained in all maximal cliques $M_{f(p)}, M_{f(p+1)}, \ldots, M_{f(q)}$.

	Therefore, for each $i \in \{p, p+1, \ldots, q\}$, both $\cliq_i \subseteq M_{f(i)}$ and $v \in M_{f(i)}$ hold, and hence $\cliq_i \cup \{v\}$ forms a clique which is contained in $M_{f(i)}$.
	Furthermore, $\cliq_{i} \cup \{v\} \onestep \cliq_{i+1} \cup \{v\}$ under $\TAR{k}$ for each $i \in \{p, p+1, \ldots, q-1\}$, because $\cliq_{i} \onestep \cliq_{i+1}$ under $\TAR{k}$.
	Recall that $\cliq_{p+1} \cup \{v\} = \cliq_{p}$ and $\cliq_{q-1} \cup \{v\} = \cliq_{q}$, and hence we replace the sub-sequence $\langle \cliq_{p}, \cliq_{p+1}, \ldots, \cliq_{q} \rangle$ of length $q-p$ with the following sequence of length $q-p-2$:
	\[
		\langle \cliq_{p+1} \cup \{v\},\ \cliq_{p+2} \cup \{v\},\ \ldots, \cliq_{q-1} \cup \{v\} \rangle. 
	\]
	However, this contradicts the assumption that $\langle \cliq_{0}, \cliq_1, \ldots, \cliq_{\ell}\rangle$ is shortest.
	\qed
\end{proof}
\medskip

\noindent
	{\bf Proof of Lemma~\ref{lem:chordal->interval}.}

	We first add two dummy vertices $\dumo$ and $\dumt$ to a given chordal graph $G$.
	We then join $\dumo$ with all vertices in $\cliq_{\ini}$ by adding new edges to $G$;
similarly, we join $\dumt$ with all vertices in $\cliq_{\tar}$.
	Let $G^\prime$ be the resulting graph.
	Then, $G^\prime$ is also a chordal graph, because the dummy vertices cannot create any new induced cycle of length more than three.
	Note that each of $\cliq_{\ini} \cup \{\dumo\}$ and $\cliq_{\tar} \cup \{\dumt\}$ forms a maximal clique in $G^\prime$.
	Furthermore, in the set $\Mset{G^\prime}$ of all maximal cliques in $G$, the only maximal cliques $\cliq_{\ini} \cup \{\dumo\}$ and $\cliq_{\tar} \cup \{\dumt\}$ contain $\dumo$ and $\dumt$, respectively. 
	
	We now construct a clique tree $\mathcal{T}^\prime$ of $G^\prime$ in linear time~\cite[\S 15.1]{Spinrad03}.
	Then, $\mathcal{T}^\prime$ contains two nodes $M_0 = \cliq_{\ini} \cup \{\dumo\}$ and $M_t = \cliq_{\tar} \cup \{\dumt\}$.
	Therefore, we can find the path $(M_{0}, M_1, \ldots, M_{t})$ in $\mathcal{T}^\prime$ in linear time.
	Let $H^{\prime \prime}$ be the subgraph of $G$ induced by the maximal cliques $M_{0}, M_1, \ldots, M_{t}$.
	Then, $H^{\prime \prime}$ is an interval graph. 
	Furthermore, since $M_0 = \cliq_{\ini} \cup \{\dumo\}$ and $M_t = \cliq_{\tar} \cup \{\dumt\}$, Lemma~\ref{lem:consecutive-cliques} implies that 
	\begin{equation} \label{eq:Hpp=Gp}
		\distTARG{H^{\prime \prime}}{\cliq_{\ini}}{\cliq_{\tar}}{k} = \distTARG{G^\prime}{\cliq_{\ini}}{\cliq_{\tar}}{k}.
	\end{equation}
	Let $\subH$ be the graph obtained from $H^{\prime \prime}$ by removing the dummy vertices $\dumo$ and $\dumt$.
	Since $H^{\prime \prime}$ is an interval graph, $H$ is also an interval graph.
	In this way, $\subH$ can be constructed in linear time. 
	
	Now we claim that 
	\begin{equation} \label{eq:Gp=G}
		\distTARG{G^\prime}{\cliq_{\ini}}{\cliq_{\tar}}{k} = \distTARG{G}{\cliq_{\ini}}{\cliq_{\tar}}{k}
	\end{equation}
and
	\begin{equation} \label{eq:Hpp=H}
		\distTARG{\subH}{\cliq_{\ini}}{\cliq_{\tar}}{k} = \distTARG{H^{\prime \prime}}{\cliq_{\ini}}{\cliq_{\tar}}{k}.
	\end{equation}
	Then, by Eqs.~(\ref{eq:Hpp=Gp})--(\ref{eq:Hpp=H}) we have $\distTARG{\subH}{\cliq_{\ini}}{\cliq_{\tar}}{k} = \distTARG{G}{\cliq_{\ini}}{\cliq_{\tar}}{k}$, as required.
	Note that $\symdiff{V(G)}{V(G^\prime)} = \symdiff{V(\subH)}{V(H^{\prime\prime})} = \{\dumo, \dumt\}$.
	Thus, to prove Eqs.~(\ref{eq:Gp=G}) and (\ref{eq:Hpp=H}), it suffices to show that there is a shortest $\TAR{k}$-sequence in $G^\prime$ (or in $H^{\prime \prime}$) from $\cliq_{\ini}$ to $\cliq_{\tar}$ which does not pass through any clique containing $\dumo$ or $\dumt$.   
	
	Let $\langle \cliq_{0}, \cliq_1, \ldots, \cliq_{\ell} \rangle$ be a shortest $\TAR{k}$-sequence in $G^\prime$ (or in $H^{\prime \prime}$) from $\cliq_{\ini}$ to $\cliq_{\ell} = \cliq_{\tar}$.
	Suppose for a contradiction that $\dumo \in \cliq_{i}$ holds for some $i \in \{1, 2, \ldots, \ell-1\}$.
(The proof for $\dumt$ is the same.)
	Since $\dumo \notin \cliq_{\ini} \cup \cliq_{\ell}$, Lemma~\ref{lem:never-appear-again} implies that there exists a pair of indices $l$ and $r$ in $\{1, 2, \ldots, \ell-1\}$ such that $l \le r$ and $\dumo \in \cliq_{i}$ holds for all $i \in \{l, l+1, \ldots, r\}$.
	Recall that $\cliq_{\ini} \cup \{\dumo\}$ is a maximal clique in $G^\prime$ (or in $H^{\prime \prime}$), and that no other maximal clique in $G^\prime$ (or in $H^{\prime \prime}$) contains $\dumo$.
	This implies that $\cliq_{i} \subseteq \cliq_{\ini} \cup \{\dumo\}$ for each $i \in \{l, l+1, \ldots, r\}$.
	Since $\cliq_{l-1} = \cliq_{l} \setminus \{\dumo\}$ and $\cliq_{r+1} = \cliq_{r} \setminus \{\dumo\}$, it follows that $\cliq_{l-1} \cup \cliq_{r+1} \subseteq \cliq_{\ini}$ and hence $\cliq_{l-1} \cup \cliq_{r+1}$ forms a clique.
	Now, by Lemma~\ref{lem:tar-dist-in-a-clique} every shortest $\TAR{k}$-sequence from $\cliq_{l-1}$ to $\cliq_{r+1}$ passes through cliques consisting of vertices only in $\cliq_{l-1} \cup \cliq_{r+1} \subseteq \cliq_{\ini}$.
	Since $\dumo \not\in \cliq_{\ini}$, this contradicts the assumption that $\langle \cliq_{l-1}, \cliq_{l}, \ldots, \cliq_{r+1} \rangle$ is shortest.
\qed

	\subsection{Correctness of the algorithm for interval graphs} \label{app:interval-algo}
	
	In this subsection, we prove the correctness of the greedy algorithm in Section~\ref{subsec:interval} and estimate its running time.
	For a vertex $v$ in a graph $G$, let $N(v) = \{w \in V(G) \mid vw \in E(G) \}$ and let $N[v] = N(v)\cup\{v\}$.
	We denote by $\deg(v)$ the degree of $v$, that is, $\deg(v) = |N(v)|$.
\medskip

	We first prove the correctness of Step~(1) of the algorithm:
if $\cliq_{\ini} \not\subseteq \cliq_{\tar}$ and $|\cliq_{\ini}| \ge k+1$, then remove a vertex $u$ with the minimum $r$-value in $\cliq_{\ini} \setminus \cliq_{\tar}$ from $\cliq_{\ini}$.
	The following lemma ensures that this operation preserves the shortest length of reconfiguration sequences.
	\begin{lemma} \label{lem:greedily-removing-a-shortest-interval}
	Suppose that $\cliq_{\ini} \not\subseteq \cliq_{\tar}$ and $|\cliq_{\ini}| \ge k+1$.
	Let $u$ be any vertex with the minimum $r$-value in $\cliq_{\ini} \setminus \cliq_{\tar}$.
	Then, 
	\[
		\distTAR{\cliq_{\ini}}{\cliq_{\tar}}{k} = \distTAR{\cliq_{\ini} \setminus \{u\}}{\cliq_{\tar}}{k} + 1.
	\]
	\end{lemma}
	\begin{proof}
	First, observe that $r_{u} = 0$ since $\cliq_{\ini} \not\subseteq M_{1}$.
	Thus, $N[u] = M_0 \subseteq N[v]$ holds for every vertex $v \in M_{0}$.
	Consider any clique $\cliq$ in $\intH$ such that $\cliq_{\ini} \onestep \cliq$ under $\TAR{k}$.
	Then, either (\one) $\cliq = \cliq_{\ini} \setminus \{ v\}$ for some vertex $v \in \cliq_{\ini}$, or (\two) $\cliq = \cliq_{\ini} \cup \{w\}$ for some vertex $w \in M_0 \setminus \cliq_{\ini}$;
recall that $\cliq_{\ini} \subseteq M_0$ and $\cliq_{\ini} \not\subseteq M_{1}$. 
	Therefore, it suffices to verify the following two inequalities:
	\begin{equation} \label{eq:delete}
		\distTAR{\cliq_{\ini} \setminus \{u\}}{\cliq_{\tar}}{k} \le \distTAR{\cliq_{\ini} \setminus \{v\}}{\cliq_{\tar}}{k}
	\end{equation}
for any vertex $v \in \cliq_{\ini}$; and 
	\begin{equation} \label{eq:add}
		\distTAR{\cliq_{\ini} \setminus \{u\}}{\cliq_{\tar}}{k} \le \distTAR{\cliq_{\ini} \cup \{w\}}{\cliq_{\tar}}{k}
	\end{equation}
for any vertex $w \in M_{0} \setminus \cliq_{\ini}$.
\medskip

	We first prove Eq.~(\ref{eq:delete}). 
	Let $v$ be any vertex in $\cliq_{\ini} \setminus \{u\}$, and let $\langle \cliq_{1}, \cliq_{2}, \ldots, \cliq_{\ell} \rangle$ be a shortest $\TAR{k}$-sequence from $\cliq_{1} = \cliq_{\ini} \setminus \{v\}$ to  $\cliq_{\ell} = \cliq_{\tar}$.
	By Lemma~\ref{lem:never-appear-again} we have $v \notin \cliq_{i}$ for all $i \in \{1,2, \ldots, \ell\}$.
	On the other hand, since $u \in \cliq_{1} \setminus \cliq_{\ell}$, there exists an index $j \ge 1$ such that $\cliq_{j+1} = \cliq_j \setminus \{ u\}$;
Lemma~\ref{lem:never-appear-again} implies that $u \in \cliq_{i}$ if and only if $i \in \{1,2,\ldots, j\}$.
	Then, $\cliq_i^\prime = (\cliq_{i} \setminus \{ u \}) \cup \{v\}$ forms a clique for each $i \in \{1,2,\ldots, j\}$, because $N[u] \subseteq N[v]$ for the vertex $v \in \cliq_{\ini} \setminus \{u\} \subset M_0$.
	For each $i \in \{1,2, \ldots, j\}$, we replace the clique $\cliq_i$ in $\langle \cliq_{1}, \cliq_{2}, \ldots, \cliq_{\ell} \rangle$ with the clique $\cliq_i^\prime = (\cliq_{i} \setminus \{ u \}) \cup \{v\}$, and obtain the following sequence $\mathcal{C}^\prime$ of cliques:
	\[
		\mathcal{C}^\prime = \langle \cliq_1^\prime,\ \cliq_2^\prime,\ \ldots, \cliq_j^\prime,\ \cliq_{j+1},\ \cliq_{j+2},\ \ldots,\  \cliq_{\ell} \rangle.  
	\]
	Since $\langle \cliq_{1}, \cliq_{2}, \ldots, \cliq_{\ell} \rangle$ is a $\TAR{k}$-sequence, we have $|\cliq_{i}^\prime| = |\cliq_{i}| \ge k$.
	Furthermore, $\cliq_{i}^\prime \onestep \cliq_{i+1}^\prime$ under $\TAR{k}$ for all $i \in \{1,2,\ldots, j-1\}$, since $\cliq_i \onestep \cliq_{i+1}$ under $\TAR{k}$.
	Finally, since $\cliq_{j+1} = \cliq_{j} \setminus \{u\}$, we have $\cliq_{j}^\prime \setminus \{v\} = \cliq_{j+1}$ and hence $\cliq_{j}^\prime \onestep \cliq_{j+1}$ under $\TAR{k}$.
	Therefore, $\mathcal{C}^\prime$ is a $\TAR{k}$-sequence from $\cliq_{1}^\prime = \cliq_{\ini} \setminus \{u\}$ to  $\cliq_{\ell} = \cliq_{\tar}$, which has the same length $\ell$ as the shortest $\TAR{k}$-sequence $\langle \cliq_{1}, \cliq_{2}, \ldots, \cliq_{\ell} \rangle$ from $\cliq_{1} = \cliq_{\ini} \setminus \{v\}$ to  $\cliq_{\ell} = \cliq_{\tar}$.
	We have thus verified Eq.~(\ref{eq:delete}). 
\medskip  

 	We then prove Eq.~(\ref{eq:add}). 
	Let $w$ be any vertex in $M_{0} \setminus \cliq_{\ini}$, and let $\langle \cliq_{1}, \cliq_2, \ldots, \cliq_{\ell} \rangle$ be a shortest $\TAR{k}$-sequence from $\cliq_{1} = \cliq_{\ini} \cup \{w\}$ to  $\cliq_{\ell} = \cliq_{\tar}$.
	Let $j \in \{1,2,\ldots, \ell-1\}$ be the index such that $u \in \cliq_{i}$ if and only if $i \in \{1,2, \ldots, j\}$.
	Since $r_{u} = 0$, all cliques $\cliq_1, \cliq_2, \ldots, \cliq_j$ are contained in $M_0$. 
	Furthermore, since $\cliq_{j+1} = \cliq_{j} \setminus \{u\}$, we have $\cliq_{j+1} \subseteq M_0$ and hence $\cliq_{1} \cup \cliq_{j+1}$ $(\subseteq M_{0})$ forms a clique.
	Then, Lemma~\ref{lem:tar-dist-in-a-clique} implies that $\distTAR{\cliq_{1}}{\cliq_{j+1}}{k} = |\symdiff{\cliq_{1}}{\cliq_{j+1}}|$.
	Note that, since the sub-sequence $\langle \cliq_1, \cliq_2, \ldots, \cliq_{j+1} \rangle$ is shortest, we have $\distTAR{\cliq_{1}}{\cliq_{j+1}}{k} = |\symdiff{\cliq_{1}}{\cliq_{j+1}}| = j$.
	On the other hand, consider the clique $\cliq_{1}^\prime = \cliq_\ini \setminus \{u\}$;
note that, since $|\cliq_{\ini}| \ge k+1$, we have $|\cliq_{1}^\prime| \ge k$. 
	Since $\cliq_1^\prime, \cliq_{j+1} \subseteq M_0$, the set $\cliq_{1}^\prime \cup \cliq_{j+1}$ forms a clique.
	Then, Lemma~\ref{lem:tar-dist-in-a-clique} implies that $\distTAR{\cliq_{1}^\prime}{\cliq_{j+1}}{k} = |\symdiff{\cliq_{1}^\prime}{\cliq_{j+1}}|$.
	We now prove that 
	\begin{equation} \label{eq:add-one}
		\distTAR{\cliq_{1}^\prime}{\cliq_{j+1}}{k} \le \distTAR{\cliq_{1}}{\cliq_{j+1}}{k} = j.
	\end{equation}
	Indeed, we show that $|\symdiff{\cliq_{1}^\prime}{\cliq_{j+1}}| \le |\symdiff{\cliq_{1}}{\cliq_{j+1}}|$, as follows.
	Since $\cliq_{1}^\prime = \cliq_{1} \setminus \{u,w\}$, $u, w \in \cliq_{1}$ and $u \notin \cliq_{j+1}$, we have 
	\begin{eqnarray*}
		|\symdiff{\cliq_{1}^\prime}{\cliq_{j+1}}| &=& |\cliq_{1}^\prime \cup \cliq_{j+1}| - |\cliq_{1}^\prime \cap \cliq_{j+1}| \\
				&=& |(\cliq_{1} \setminus \{u,w\}) \cup \cliq_{j+1}| - |(\cliq_{1} \setminus \{u,w\}) \cap \cliq_{j+1}| \\
				&=& \begin{cases}
						(|\cliq_{1}  \cup \cliq_{j+1}| - |\{u\}|) - (|\cliq_{1} \cap \cliq_{j+1}| - |\{w\}|) & \text{if } w \in \cliq_{j+1}, \\
						(|\cliq_{1} \cup \cliq_{j+1}| - |\{u,w\}|) - |\cliq_{1} \cap \cliq_{j+1}| & \text{if } w \notin \cliq_{j+1}
						\end{cases} \\
				&\le& |\cliq_{1} \cup \cliq_{j+1}| - |\cliq_{1} \cap \cliq_{j+1}| \\
				&=& |\symdiff{\cliq_{1}}{\cliq_{j+1}}|.
	\end{eqnarray*}
	Let $\langle \cliq_1^\prime, \cliq_2^\prime, \ldots, \cliq_{j}^\prime, \cliq_{j+1} \rangle$ be a shortest $\TAR{k}$-sequence from $\cliq_{1}^\prime$ to $\cliq_{j+1}$.
	Then, by Eq.~(\ref{eq:add-one}) the length of $\langle \cliq_1^\prime, \cliq_2^\prime, \ldots, \cliq_{j}^\prime, \cliq_{j+1} \rangle$ is at most $j$.
	We replace the sub-sequence $\langle \cliq_1, \cliq_2, \ldots, \cliq_{j}, \cliq_{j+1} \rangle$ of length $j$ with the $\TAR{k}$-sequence $\langle \cliq_1^\prime, \cliq_2^\prime, \ldots, \cliq_{j}^\prime, \cliq_{j+1} \rangle$.
	Then, $\langle \cliq_1^\prime, \cliq_2^\prime, \ldots, \cliq_{j}^\prime, \cliq_{j+1}, \cliq_{j+2}, \ldots, \cliq_{\ell} \rangle$ is a $\TAR{k}$-sequence from $\cliq_{1}^\prime = \cliq_\ini \setminus \{u\}$ to $\cliq_{\ell} = \cliq_{\tar}$, whose length is at most $\ell-1 = \distTAR{\cliq_{\ini} \cup \{w\}}{\cliq_{\tar}}{k}$.
	We have thus verified Eq.~(\ref{eq:add}). 
\qed
\end{proof}

	We then prove the correctness of Step~(2) of the algorithm:
if no vertex can be deleted from $\cliq_{\ini}$ according to Lemma~\ref{lem:greedily-removing-a-shortest-interval}, then add a vertex $u$ chosen by the following lemma, with preserving the shortest length of reconfiguration sequences.
	\begin{lemma} \label{lem:greedily-taking-a-longest-interval}
	Assume that $\cliq_{\ini} \subseteq \cliq_{\tar}$ or $|\cliq_{\ini}| = k$.
	Let $u$ be any vertex in $(\cliq_{\tar} \setminus \cliq_{\ini}) \cap M_{0}$ if exists{\rm ;} 
otherwise, let $u$ be any vertex with the maximum $r$-value in $M_{0} \setminus \cliq_{\ini}$.
	Then,
	\[
		\distTAR{\cliq_{\ini}}{\cliq_{\tar}}{k} = \distTAR{\cliq_{\ini} \cup \{u\}}{\cliq_{\tar}}{k} + 1.
	\]
	\end{lemma}
	\begin{proof}
	Note that, if $|\cliq_{\ini}| = k$, then no vertex can be deleted from $\cliq_{\ini}$ due to the size constraint $k$.
	On the other hand, if $\cliq_{\ini} \subseteq \cliq_{\tar}$, then by Lemma~\ref{lem:never-appear-again} no shortest $\TAR{k}$-sequence from $\cliq_{\ini}$ to $\cliq_{\tar}$ deletes any vertex $v$ in $\cliq_{\ini}$, because $v \in \cliq_{\ini} \cap \cliq_{\tar}$.
	Therefore, in any shortest $\TAR{k}$-sequence $\langle \cliq_{\ini}, \cliq_{1}, \ldots, \cliq_{\ell} \rangle$ from $\cliq_{\ini}$ to $\cliq_{\ell} = \cliq_{\tar}$, the clique $\cliq_1$ must be obtained from $\cliq_{\ini}$ by adding a vertex $v \in V(G) \setminus \cliq_{\ini}$.
	Furthermore, since $\cliq_{\ini} \subseteq M_0$, $\cliq_{\ini} \not\subseteq M_1$ and $\cliq_1 = \cliq_{\ini} \cup \{v\}$ is a clique, the added vertex $v$ must be in $M_0 \setminus \cliq_{\ini}$. 
	Thus, to prove the lemma, it suffices to show that
	\begin{equation} \label{eq:two-one}
		\distTAR{\cliq_{\ini} \cup \{u\}}{\cliq_{\tar}}{k} \le \distTAR{\cliq_{\ini} \cup \{v\}}{\cliq_{\tar}}{k}
	\end{equation}
for any vertex $v \in M_{0} \setminus \cliq_{\ini}$. 

	Let  $v$ be any vertex in $M_{0} \setminus \cliq_{\ini}$, and let $\langle \cliq_{1}, \cliq_2, \ldots, \cliq_{\ell} \rangle$ be a shortest $\TAR{k}$-sequence from $\cliq_{1} = \cliq_{\ini} \cup \{v\}$ to $\cliq_{\ell} = \cliq_{\tar}$.
	For each $i \in \{1,2,\ldots, \ell\}$, let 
	\begin{equation} \label{eq:two-two}
		\cliq_{i}^\prime =
				\begin{cases}
					(\cliq_{i} \setminus \{v\}) \cup \{u\} & \text{if } v \in \cliq_{i} \text{ and } u \notin \cliq_{i},\\
					\cliq_{i} & \text{otherwise}.
				\end{cases}
	\end{equation}
	We will prove below that $\langle \cliq_{1}^\prime, \cliq_2^\prime, \ldots, \cliq_{\ell}^\prime \rangle$ is a $\TAR{k}$-sequence from $\cliq_{\ini} \cup \{u\}$ to $\cliq_{\tar}$.
	Then, since $\langle \cliq_{1}^\prime, \cliq_2^\prime, \ldots, \cliq_{\ell}^\prime \rangle$ is of length $\ell = \distTAR{\cliq_{\ini} \cup \{v\}}{\cliq_{\tar}}{k}$, Eq.~(\ref{eq:two-one}) follows. 
	
	We first claim that $\cliq_{1}^\prime = \cliq_{\ini} \cup \{u\}$ and $\cliq_{\ell}^\prime = \cliq_{\ell}$. 
	Since $v \in \cliq_{\ini} \cup \{v\} = \cliq_1$ and $u \not\in \cliq_{\ini} \cup \{v\} = \cliq_1$, we have $\cliq_{1}^\prime = (\cliq_1 \setminus \{v\}) \cup \{u\} = \cliq_{\ini} \cup \{u\}$. 
	On the other hand, if $u$ is chosen from $(\cliq_{\tar} \setminus \cliq_{\ini}) \cap M_{0}$, then $u \in \cliq_{\tar} = \cliq_{\ell}$ and hence $\cliq_{\ell}^\prime = \cliq_{\ell}$.
	Otherwise, $(\cliq_{\tar} \setminus \cliq_{\ini}) \cap M_{0} = (M_{0} \setminus \cliq_{\ini}) \cap \cliq_{\tar} = \emptyset$ holds, and hence $v \in M_{0} \setminus \cliq_{\ini}$ is not contained in $\cliq_{\tar} = \cliq_{\ell}$;
we then have $\cliq_{\ell}^\prime = \cliq_{\ell}$.

	We then prove that $\cliq_i^\prime$ forms a clique of size at least $k$ for each $i \in \{1,2,\ldots,\ell\}$, and prove that $\cliq_{i}^\prime \onestep \cliq_{i+1}^\prime$ under $\TAR{k}$ for each $i \in \{1,2,\ldots, \ell-1\}$. 
	Since $\langle \cliq_{1}, \cliq_2, \ldots, \cliq_{\ell} \rangle$ is a $\TAR{k}$-sequence, by Eq.~(\ref{eq:two-two}) we have $|\cliq_i^\prime| = |\cliq_i| \ge k$. 
	Therefore, it suffices to show that $\cliq_{i}^\prime \cup \cliq_{i+1}^\prime$ forms a clique such that $|\symdiff{\cliq_{i}^\prime}{\cliq_{i+1}^\prime}| = 1$ for each $i \in \{1,2,\ldots, \ell-1\}$.
	This claim trivially holds for the case where both $\cliq_i^\prime = \cliq_i$ and $\cliq_{i+1}^\prime = \cliq_{i+1}$ hold, because $\langle \cliq_{1}, \cliq_2, \ldots, \cliq_{\ell} \rangle$ is a $\TAR{k}$-sequence.
	By symmetry, we thus assume that $\cliq_{i}^\prime = (\cliq_{i} \setminus \{v\}) \cup \{u\}$, that is, both $v \in \cliq_{i}$ and $u \notin \cliq_{i}$ hold.
	Then, there are the following three cases to consider;
note that, since both $v \in \cliq_{i}$ and $u \notin \cliq_{i}$ hold and $\cliq_i \onestep \cliq_{i+1}$ under $\TAR{k}$, we do not need to consider the case where both $v \not\in \cliq_{i+1}$ and $u \in \cliq_{i+1}$ hold.
\medskip

\noindent
	{\bf Case (\one) $v \in \cliq_{i+1}$ and $u \notin \cliq_{i+1}$.}
	
	In this case, we have $\cliq_{i+1}^\prime = (\cliq_{i+1} \setminus \{v\}) \cup \{u\}$. 
	Since $\cliq_{i}^\prime = (\cliq_{i} \setminus \{v\}) \cup \{u\}$ and $\cliq_i \onestep \cliq_{i+1}$ under $\TAR{k}$, we have $|\symdiff{\cliq_{i}^\prime}{\cliq_{i+1}^\prime}| = |\symdiff{\cliq_i}{\cliq_{i+1}}| = 1$.
	Notice that $l_v = l_u = 0$ and $r_v \le r_u$, because $r_u = t$ or $u$ has the maximum $r$-value in $M_{0} \setminus \cliq_{\ini}$. 
	Therefore, $N[v] \subseteq N[u]$ holds.
	Then, since $\cliq_{i} \cup \cliq_{i+1}$ is a clique, $\cliq_{i}^\prime \cup \cliq_{i+1}^\prime = \bigl( (\cliq_{i} \cup \cliq_{i+1}) \setminus \{v\} \bigr) \cup \{u\}$ forms a clique.
\medskip

\noindent
	{\bf Case (\two) $v, u \in \cliq_{i+1}$.} 

	In this case, we have $\cliq_{i+1}^\prime = \cliq_{i+1}$. 
	Recall that both $v \in \cliq_{i}$ and $u \notin \cliq_{i}$ hold. 
	Then, since $v, u \in \cliq_{i+1}$ and $\cliq_i \onestep \cliq_{i+1}$ under $\TAR{k}$, we have $\cliq_i \cup \{ u\} = \cliq_{i+1} = \cliq_{i+1}^\prime$.
	Since $\cliq_{i}^\prime = (\cliq_{i} \cup \{u\}) \setminus \{v\}$, we thus have $\cliq_{i}^\prime = \cliq_{i+1}^\prime \setminus \{v\}$ and hence $|\symdiff{\cliq_{i}^\prime}{\cliq_{i+1}^\prime}| = |\{v\}| = 1$. 
	Furthermore, since $\cliq_{i+1}^\prime = \cliq_{i+1}$ and $\cliq_{i+1}$ is a clique, $\cliq_{i}^\prime \cup \cliq_{i+1}^\prime = \cliq_{i+1}^\prime$ forms a clique.
\medskip

\noindent
	{\bf Case (\three) $v, u \notin \cliq_{i+1}$.} 
	
	In this case, we have $\cliq_{i+1}^\prime = \cliq_{i+1}$. 
	Recall again that both $v \in \cliq_{i}$ and $u \notin \cliq_{i}$ hold. 
	Then, since $v, u \not\in \cliq_{i+1}$ and $\cliq_i \onestep \cliq_{i+1}$ under $\TAR{k}$, we have $\cliq_i \setminus \{ v\} = \cliq_{i+1} = \cliq_{i+1}^\prime$.
	Since $\cliq_{i}^\prime = (\cliq_{i} \setminus \{v\}) \cup \{u\}$, we thus have $\cliq_{i}^\prime = \cliq_{i+1}^\prime \cup \{u\}$ and hence $|\symdiff{\cliq_{i}^\prime}{\cliq_{i+1}^\prime}| = |\{u\}| = 1$. 
	Then, $\cliq_{i}^\prime \cup \cliq_{i+1}^\prime = \cliq_{i}^\prime = (\cliq_{i} \setminus \{v\}) \cup \{u\}$.
	Since $N[v] \subseteq N[u]$ holds and $\cliq_i$ is a clique, $\cliq_{i}^\prime \cup \cliq_{i+1}^\prime = (\cliq_{i} \setminus \{v\}) \cup \{u\}$ forms a clique.
\medskip

	In this way, we have proved that $\langle \cliq_{1}^\prime, \cliq_2^\prime, \ldots, \cliq_{\ell}^\prime \rangle$ is a $\TAR{k}$-sequence from $\cliq_{\ini} \cup \{u\}$ to $\cliq_{\tar}$, and hence Eq.~(\ref{eq:two-one}) holds as we have mentioned above. 
\qed
\end{proof}

	The correctness of the greedy algorithm in Section~\ref{subsec:interval} follows from Lemmas~\ref{lem:greedily-removing-a-shortest-interval} and \ref{lem:greedily-taking-a-longest-interval}.
	Therefore, to complete the proof of Theorem~\ref{the:chordal}, we now show that the algorithm runs in linear time. 
\medskip

\noindent
	{\bf Estimation of the running time.}

	Lemma~\ref{lem:never-appear-again} implies that each vertex is removed at most once and added at most once in any shortest $\TAR{k}$-sequence.
	Therefore, it suffices to show that each removal and addition of a vertex $u$ can be done in time $O(\deg(u))$, because $\sum_{u \in V(G)} \deg(u) = 2 |E(\intH)|$.

	We first estimate the running time for Step~(1) of the algorithm. 
	We first check whether both $\cliq_{\ini} \not\subseteq \cliq_{\tar}$ and $|\cliq_{\ini}| \ge k+1$ hold or not. 
	These conditions can be checked in constant time by maintaining $|\cliq_{\ini}|$ and $|\cliq_{\ini} \cap \cliq_{\tar}|$.
	We then find a vertex $u$ with the minimum $r$-value in $\cliq_{\ini} \setminus \cliq_{\tar}$; this can be done in time $O(|\cliq_{\ini}|)$.
	After the removal of $u$, the clique $C_{\ini}:=C_{\ini} \setminus \{u\}$ may be included by some of $M_{1}, M_{2}, \ldots, M_{t}$; 
in such a case, we need to shift the indices of $M_i$ so that $\cliq_{\ini} \subseteq M_{0}$ and $\cliq_{\ini} \not\subseteq M_{1}$ hold. 
	To do so, we compute the shift-value $i_{0} = \min\{r_{u} \mid u \in \cliq_{\ini}\}$, and set $M_{i} := M_{i - i_{0}}$ for each $i \in \{1,2,\ldots, t\}$ and $r_{w} := r_{w} - i_{0}$ for each vertex $w \in V(\intH)$.
	However, since we just have to compute and store only the shift-value $i_{0}$ in the actual process, this post-process can be done also in time $O(|\cliq_{\ini}|)$.
	Since $\cliq_{\ini} \subseteq N[u]$, we have $|\cliq_{\ini}| \le \deg(u) + 1$.
	Therefore, Step~(1) can be executed in time $O(\deg(u))$.

	We then estimate the running time for Step~(2) of the algorithm. 
	We find a vertex $u$ which either is in $(\cliq_{\tar} \setminus \cliq_{\ini}) \cap M_{0}$ or has the maximum $r$-value in $M_{0} \setminus \cliq_{\ini}$.
	In either case, such a vertex $u$ can be found in time $O(|M_{0}|)$.
	Since $M_{0} \subseteq N[u]$, the addition of $u$ can be done in time $O(\deg(u))$.